\documentclass[10pt,twocolumn,twoside]{IEEEtran}

\usepackage{graphicx}
\usepackage{amsfonts,amsthm,color}
\usepackage{mathrsfs,amsmath,arydshln,latexsym}
\usepackage{amssymb}
\usepackage{cite}
\usepackage{url}
\usepackage[bookmarks,colorlinks]{hyperref}
\usepackage[linesnumbered,ruled,vlined]{algorithm2e}
\usepackage{multirow}
\usepackage{stfloats}
\usepackage{balance}

\allowdisplaybreaks[4]

\newtheorem{theorem}{Theorem}
\newtheorem{proposition}{Proposition}

\begin{document}

\title{Compressive Subspace Learning with Antenna Cross-correlations for Wideband Spectrum Sensing}

\author{Tierui~Gong, 
		Zhijia~Yang,
        Meng~Zheng,
        Zhifeng~Liu,
        and Gengshan Wang
% <-this % stops a space
\vspace{-0.36cm}
\thanks{
	This work was supported in part by the National Key Research and Development Program of China under grant 2017YFA0700300, in part by the National Natural Science Foundation of China under grant 61673371, in part by the  International Partnership Program of Chinese Academy of Sciences under grant 173321KYSB20180020, and in part by the Liaoning Provincial Natural Science Foundation of China under grant 2019-YQ-09.
}% <-this % stops a space
\thanks{
	T. Gong and G. Wang are with the State Key Laboratory of Robotics, the Key Laboratory of Networked Control Systems, Shenyang Institute of Automation, Chinese Academy of Sciences, Shenyang 110016, China, the Institutes for Robotics and Intelligent Manufacturing, Chinese Academy of Sciences, Shenyang 110169, China, and also with the University of Chinese Academy of Sciences, Beijing 100049, China. (e-mail: gongtierui@sia.cn; wanggengshan@sia.cn).
	
	Z. Yang, M. Zheng, and Z. Liu are with the State Key Laboratory of Robotics, the Key Laboratory of Networked Control Systems, Shenyang Institute of Automation, Chinese Academy of Sciences, Shenyang 110016, China, and the Institutes for Robotics and Intelligent Manufacturing, Chinese Academy of Sciences, Shenyang 110169, China. (e-mail: yang@sia.ac.cn; zhengmeng\_6@sia.cn; liuzhifeng@sia.cn).}
}




\maketitle

\begin{abstract}
Compressive subspace learning (CSL) with the exploitation of space diversity has found a potential performance improvement for wideband spectrum sensing (WBSS). However, previous works mainly focus on either exploiting antenna auto-correlations or adopting a multiple-input multiple-output (MIMO) channel without considering the spatial correlations, which will degrade their performances. In this paper, we consider a spatially correlated MIMO channel and propose two CSL algorithms (i.e., mCSLSACC and vCSLACC) which exploit antenna cross-correlations, where the mCSLSACC utilizes an antenna averaging temporal decomposition, and the vCSLACC uses a spatial-temporal joint decomposition. For both algorithms, the conditions of statistical covariance matrices (SCMs) without noise corruption are derived. Through establishing the singular value relation of SCMs in statistical sense between the proposed and traditional CSL algorithms, we show the superiority of the proposed CSL algorithms. By further depicting the receiving correlation matrix of MIMO channel with the exponential correlation model, we give important closed-form expressions for the proposed CSL algorithms in terms of the amplification of singular values over traditional CSL algorithms. Such expressions provide a possibility to determine optimal algorithm parameters for high system performances in an analytical way. Simulations validate the correctness of this work and its performance improvement over existing works in terms of WBSS performance.
\end{abstract}


% \fi
\IEEEpeerreviewmaketitle

%===================================================================%
%--------------------------Introduction-----------------------------%
%===================================================================%
%\newpage
%\vspace{-1em}
\section{Introduction}
\IEEEPARstart{T}he demand of a large amount of spectrum resources has been put forward with the tremendous growth of wireless devices and services. However, spectrum access opportunities are limited for new wireless devices because most of the available spectrum resources have been allocated to primary users (PUs) according to a static spectrum allocation policy. A contradictory reality is that the licensed spectrum resources are usually under-utilized by PUs \cite{Kolodzy2002Spectrum, Mchenry2005NSF, Chen2016Survey}, which not only wastes the valuable spectrum resources but also arises a spectrum shortage dilemma for new wireless devices in modern and future wireless communications.
Cognitive radio (CR) \cite{Haykin2005Cognitive}, attracting great research interests in recent years, has been considered as one of the promising techniques to deal with the spectrum shortage dilemma.  When spectrum resources are released by PUs, CR allows secondary users (SUs) to access the licensed spectrum in an opportunistic fashion without introducing any interference to PUs.

As one of the most critical functions to enable CRs, spectrum sensing is used to reliably detect PU activities on the licensed spectrum. 
Traditional narrowband spectrum sensing (NBSS) techniques \cite{Yucek2009Survey,Zheng2016SMCSS,Zheng2017Energy}, relying on the Nyquist sampling, generally fail to satisfy the requirement of CRs to sense a broad frequency range, which is known as the wideband spectrum sensing (WBSS) \cite{Hattab2014Multiband, Ali2017Advances, Hamdaoui2018Compressed,Qin2018Sparse}. 
Benefiting from the compressed sensing theory \cite{Candes2006Robust}, WBSS techniques based on the sub-Nyquist sampling \cite{Mishali2011Sub}
have received significant interests recently. By exploiting the spectrum sparsity, some works apply sparse recovery to obtain the original spectrum or the signals by solving an under-determined linear equation \cite{Khalfi2018Efficient, Qi2018Blind, Zhang2018Autonomous, Wang2018Phased}. Alternatively, the other works recover the power spectral density (PSD) of original signals using least-squares by solving an over-determined linear equation, which relaxes the prior sparsity assumption \cite{Ariananda2012Compressive,Cohen2014Sub,Romero2016Compressive}. Even though advantages of the sub-Nyquist sampling based WBSS techniques have been shown \cite{Sun2013Wideband}, their performances could noticeably deteriorate in practical communication scenarios due to the wireless channel fading.

An effective approach to promote the sensing performance is via the space diversity realized by multiple-input multiple-output (MIMO) systems. The advantages of MIMO systems mainly come from the noise uncorrelation and different fading levels of multiple antennas, which can be combined to combat performance degradation caused by the low channel quality. %The MIMO technique has been widely investigated in wireless communications \cite{Marzetta2010Noncooperative,Gong2020RF}. 
The validation of sensing performance improvement by the MIMO technique has been verified in both NBSS and WBSS scenarios \cite{Taherpour2010Multiple, Yang2011Multi, Ioushua2017CaSCADE}. Additionally, in MIMO systems, spatial correlations exist between antennas due to the integration of multiple antennas in a limited space \cite{Masouros2013Large}. It is notable that spatial correlations help increase the signal similarity between antennas, which greatly contributes to the reconstruction of source signals at the receiver side conditioned on channel fading and background noise. It is thus necessary to consider and analyze such a spatial correlation effect in WBSS of MIMO systems. A recent work \cite{Gong2019Compressed} employs spatial correlations to improve WBSS. Moreover, previous works \cite{Oude2011Lowering,Kitsunezuka2015Cross,Merritt2018High} have demonstrated that the exploitation of antenna cross-correlations can lower the so-called SNR wall. Therefore, the exploitation of antenna cross-correlations in WBSS of MIMO systems has its penitential improvement in the spectrum sensing performance.

On the other hand, subspace learning has been widely investigated and used in many areas of signal processing, and also has found its application {in} performance improvement in spectrum sensing. A combination of MIMO technique and subspace learning has been studied in \cite{Zeng2009Eigenvalue, Li2016Maximum, Liu2017Optimal, Bouallegue2018Blind, Jin2019Spectrum} for the NBSS. A general procedure is with a first eigenvalue computation of the statistical covariance matrix (SCM) and a following PU detection by the obtained eigenvalues. However, these works neither consider spatial correlations of the MIMO channel nor exploit antenna cross-correlations. More important, these works \cite{Zeng2009Eigenvalue, Li2016Maximum, Liu2017Optimal, Bouallegue2018Blind, Jin2019Spectrum} are not suitable for the WBSS.

In recent {years}, compressive subspace learning (CSL) has attracted growing research interests and found its application {in} the sub-Nyquist sampling based WBSS \cite{Ma2016Reliable, Ioushua2017CaSCADE, Zhang2018Distributed}. The CSL is applied to reduce the communication overhead for cooperative WBSS in \cite{Ma2016Reliable} and the computation complexity for cognitive Internet of Things in \cite{Zhang2018Distributed}. {Improvements} in sensing performance by the CSL are verified as well. However, the CSL algorithm at each sensing node in \cite{Ma2016Reliable} does not exploit the space diversity, which will suffer a performance degradation by the wireless channel fading. The CSL algorithm in \cite{Zhang2018Distributed} is in fact an extension of the CSL algorithm in \cite{Ma2016Reliable} to multiple cognitive nodes and space diversity exploitation is however not analyzed. The work in \cite{Ioushua2017CaSCADE} introduces space diversity into CSL for WBSS and the performance improvement has been verified. The authors have considered the line of sight MIMO channel, which is, however, over-simple for wireless communication in practice. Even though exploiting the space diversity, both the mentioned works \cite{Ioushua2017CaSCADE, Zhang2018Distributed} do not consider the spatial correlations and antenna cross-correlations. The work in \cite{Gong2019Compressed} considers the spatially correlated MIMO channel and provides analyses on CSL algorithms with space diversity influence. However, the proposed CSL algorithms only exploit antenna auto-correlations and the corresponding theoretical closed form expressions are not yet provided.

Motivated by the above discussions, this paper proposes two CSL algorithms by exploiting antenna cross-correlations in the spatially correlated MIMO channel. 
We further give the conditions in which  SCMs of the proposed algorithms are not corrupted by noise. In the noise-free conditions, the proposed algorithms have advantages in improving the sensing performance. The superiority and statistical theoretical analyses of the proposed algorithms over CSL algorithms without considering spatial correlations are also given. We can take use of the achieved results to analyze how spatial correlation and the number of deployed antennas influence the performance for given antenna arrays. The analyses show that the spatial correlation indeed yields the performance gain over traditional CSL algorithms. To sum up, the main contributions are listed as follows:
	\begin{itemize}
		\item  
		New algorithms:
		We propose the mCSLSACC (matrix form CSL with the sum of antenna cross-correlations) and the vCSLACC (vector form CSL with antenna cross-correlations) for signal subspace learning by sub-samples. 
		
		\item 
		Performance analysis:  
		We establish the SCM relations based on the Kronecker model \cite{Costa2010Multiple} between the proposed and traditional CSL algorithms, and obtain the singular value relation of SCMs to describe the amplification effect on traditional CSL algorithms.
		Employing the exponential correlation model, we further derive the amplification factor of singular values by the mCSLSACC, and provide the upper and lower bounds of the maximum amplification factor for the vCSLACC.
		
		\item 
		Simulation verification: Simulations verify the correctness of obtained results and show that the proposed CSL algorithms indeed increase the sensing performance of WBSS and outperform traditional CSL algorithms
	\end{itemize}

\emph{Organization:} 
The remainder of this paper is organized as follows. System description and problem formulation are introduced in Section \ref{SectionII}. The mCSLSACC and the  vCSLACC are proposed in Section \ref{SectionIII}. Detailed performance analyses of the two proposed algorithms are investigated in Section \ref{SectionIV} and Section \ref{SectionV}, respectively. Simulations are performed in Section \ref{SectionVII} and conclusions are made in Section \ref{SectionVIII}.

\emph{Notations:} 
Throughout the paper, we use boldface letter (e.g., $\mathbf{x}$) to denote column vectors and boldface upper-case letter (e.g., $\mathbf{X}$) to represent matrices. $\mathbf{I}_{n}$ is an $n \times n$ identity matrix.  Additionally, $(\cdot)^{H}$, $(\cdot)^{T}$ and $(\cdot)^{*}$ represent the Hermitian transpose, transpose and conjugate, respectively. The $[\mathbf{X}]_{:,i:j}$ (or $[\mathbf{x}]_{i:j}$) denotes a sub-matrix of $\mathbf{X}$ (or a sub-vector of $\mathbf{x}$) constituting by columns (or elements) indexed by $i \le m \le j$. $vec(\cdot)$, $diag(\cdot)$, $blkdiag(\cdot)$,  $trace(\cdot)$ and $sv(\cdot)$ stand for a vector stacked by columns of a matrix, a diagonal matrix, a block diagonal matrix, trace of a matrix and a diagonal matrix with singular values as its diagonal elements, respectively. We use $\Vert \cdot \Vert_{F}$ and $\Vert \cdot \Vert_{0,2}$ to denote the Frobenius norm and $l_{0,2}$ norm, respectively. The set of complex number and the complex normal distribution are denoted by $\mathbb{C}$ and $\mathcal{CN}$. $\otimes$, $\mathscr{E}\{ \cdot \}$ and $\triangleq$ denote the Kronecker product, expectation and definition operators, respectively. $\cup$ and $|\cdot|$ stand for the set union and the set cardinality, respectively.

%===================================================================%
%--------------------------Signal Model-----------------------------%
%===================================================================%
\begin{figure*}[t!]
	\centering
	\includegraphics[height=3.6cm, width=18.0cm]{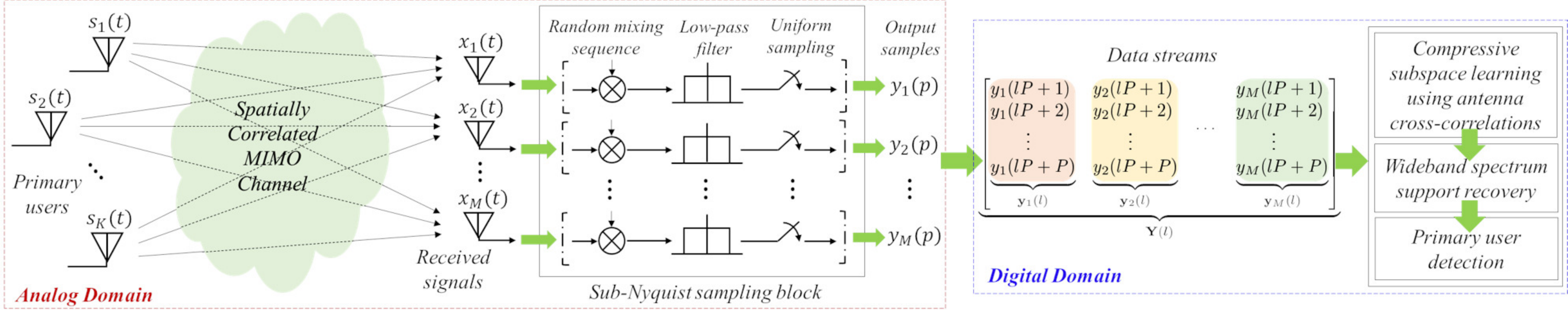}
	\vspace{-0.8em}
	\caption{The considered system architecture}
	\label{TotalSystem}
\end{figure*}

\section{System Description and Problem Formulation}
\label{SectionII}
\subsection{System Description}
Consider an $M$-antenna CR that senses a wideband spectrum utilized by $K$ PUs. Each PU is equipped with one antenna and the PU signal, denoted as $s_{k}(t)$, is scattered in multiple paths to arrive at a receiving antenna, where such multiple scattered paths are equivalent to a single path with a channel gain $g_{mk}(t)$ as shown in {the analog domain {in}} Fig. \ref{TotalSystem}. The channel gain matrix $\mathbf{G} \in \mathbb{C}^{M \times K}$ is defined and assumed to be fixed during a sensing period $T_{s}$. Without loss of generality, the received signal at the $m$th receiving antenna is denoted as $x_{m}(t)$, which is a sum of noise and $K$ faded PU signals. The noise, denoted as $n_{m}(t)$, is assumed a complex additive white Gaussian noise (AWGN) with zero mean and $\sigma^{2}_{m}$ variance, i.e., $n_{m}(t) \sim \mathcal{CN}(0,\sigma^{2}_{m})$. First, we denote $\bar{\mathbf{X}} = [\mathbf{x}_{1},\mathbf{x}_{2},\cdots,\mathbf{x}_{M}]^T \in \mathbb{C}^{M \times Q}$, $\bar{\mathbf{S}} = [\mathbf{s}_{1},\mathbf{s}_{2},\cdots,\mathbf{s}_{K}]^T \in \mathbb{C}^{K \times Q}$ and $\bar{\mathbf{N}} = [\bar{\mathbf{n}}_{1},\bar{\mathbf{n}}_{2},\cdots,\bar{\mathbf{n}}_{M}]^T \in \mathbb{C}^{M \times Q}$ as matrices with each column composed by $Q$ Nyquist samples of $x_{m}(t)$, $s_{k}(t)$ and $n_{m}(t)$, respectively. We then write the received sample matrix as
	\begin{equation}
	\begin{aligned}
		\bar{\mathbf{X}} = \mathbf{G} \bar{\mathbf{S}} + \bar{\mathbf{N}}.
		\label{CRModel}
	\end{aligned}
	\end{equation}
The PU signals $s_{k}(t) (1 \le k \le K)$ are assumed to be independent with each other. The complex AWGN $n_{m}(t) (1 \le m \le M)$ are assumed to be \emph{i.i.d.} and also independent with the PU signals. The considered channels between the CR and PUs are denoted by a spatially correlated MIMO channel model described by the Kronecker model \cite{Costa2010Multiple}, whose channel gain matrix is given by
	\begin{equation}
	\begin{aligned}
		\mathbf{G} = \mathbf{Q}^{\frac{1}{2}} \mathbf{G}_{w} \mathbf{P}^{\frac{1}{2}}, 
		\label{KroneckerModel}
	\end{aligned}
	\end{equation}
where $\mathbf{P}$ and $\mathbf{Q}$ are the transmitting and the receiving correlation matrices, respectively; $\mathbf{G}_{w}$ is a matrix constituted by \emph{i.i.d.} complex Gaussian variables with zero mean and unit variance. The independence between any two PU signals makes the transmitting correlation matrix diagonal, i.e.,  {$\mathbf{P} = diag\{ \sigma_{1}^{2}, \sigma_{2}^{2}, \cdots, \sigma_{K}^{2} \} \in \mathbb{C}^{K \times K}$} with diagonal elements {$\sigma_{k}^{2} (1 \le k \le K)$} being {transmit} powers of PUs. To further characterize the receiving correlation matrix $\mathbf{Q}$, we define the correlation coefficient between the $m_1$th antenna and $m_2$th antenna as $\rho_{m_1 m_2} \in \mathbb{C} (0 \le \left| \rho_{m_1 m_2} \right| \le 1)$. Then, we have $\mathbf{Q} = \left\{\rho_{m_1 m_2} \right\} \in \mathbb{C}^{M \times M} (1 \le m_1, m_2 \le M)$.

The received signals are passed to a multi-antenna sub-Nyquist sampling board to obtain the sub-sampled data streams. The received signal $x_{m}(t)$ of the $m$th antenna is then mixed with a random sequence, low-pass filtered, and uniformly sampled at $f_m$ to output the sub-samples $y_{m}(p)$ as shown in Fig. \ref{TotalSystem}. The hardware configurations of all branches are set the same. The whole process of sub-Nyquist sampling can be represented by a measurement matrix $\mathbf{\Omega} \in \mathbb{C}^{P \times Q}$ with $P$ being the number of sub-samples and $P \ll Q$ \cite{Tropp2009Beyond}. Let $\mathbf{Y} = [\mathbf{y}_{1},\mathbf{y}_{2},\cdots,\mathbf{y}_{M}] \in \mathbb{C}^{P \times M}$ represent the sub-sample matrix of all antennas. Then, the input-output relation can be formulated as
	\begin{equation}
	\begin{aligned}
	\mathbf{Y} = \mathbf{\Omega} \bar{\mathbf{X}}^{T}.
	\label{eqGMCmatrix}
	\end{aligned}
	\end{equation}

\subsection{Problem Formulation}
For sub-Nyquist sampling based WBSS, the wideband spectrum is assumed to be sparse with a total sensed bandwidth limited to $W$ Hz and each bandwidth of a PU signal restricted to $B$ Hz. Each antenna receives a $W$ Hz wideband signal that is a linear superposition of $K$ different PU signals and an AWGN. 
Based on \eqref{eqGMCmatrix} and the frequency domain relation $\bar{\mathbf{X}}^{T} = \mathbf{\Psi} \mathbf{Z} + \bar{\mathbf{N}}^{T}$ with $\mathbf{Z}  = [\mathbf{z}_{1}, \mathbf{z}_{2}, \cdots, \mathbf{z}_{M}]$ being the sparse Fourier coefficient matrix and $\mathbf{\Psi}$ being the discrete Fourier transform matrix, $\mathbf{Y}$ can be expressed in a sparse model as 
	\begin{equation}
	\begin{aligned}
	\mathbf{Y} = \mathbf{A} \mathbf{Z} + \mathbf{N},
	\label{eq_MAGMCmat}
	\end{aligned}
	\end{equation}
with $\mathbf{A} \triangleq \mathbf{\Omega} \mathbf{\Psi}$ and $\mathbf{N} \triangleq \mathbf{\Omega} \bar{\mathbf{N}}^{T}$. Here the AWGN after sub-Nyquist sampling is still the AWGN {yet with the larger noise power. The noise enhancement is induced by noise folding factor that {is defined as} the output-input noise power ratio of a sub-Nyquist sampling system} \cite{Gong2019Property}. We define the spectrum supports $\mathbb{S} = supp(\mathbf{Z})$ as the index set of disjoint narrowbands on which the Fourier coefficients of PU signals are lying. The spectrum supports $\mathbb{S}$ is unknown in prior due to the spectrum mobility of PUs. We aim to simultaneously detect the existing of PU signals and the corresponding spectrum support $\mathbb{S}$ from the sub-samples $\mathbf{Y}$.

Based on \eqref{CRModel} and \eqref{eqGMCmatrix}, $\mathbf{Y}$ can be reformulated as 
	\begin{equation}
	\begin{aligned}
	\mathbf{Y}= \mathbf{\Omega} \bar{\mathbf{S}}^{T} \mathbf{G}^{T} + \mathbf{N}.
	\label{eqMAGMCmat}
	\end{aligned}
	\end{equation}
As seen from \eqref{eqMAGMCmat}, the sub-samples are influenced by both the wireless channel fading and the AWGN.

Aiming to obtain clean sub-samples from the noisy ones by using CSL algorithms, we expect to achieve a better WBSS performance \cite{Ma2016Reliable}. Without prior information on the wireless channel and transmitted signals, we consider the blind recovery problem 
	\begin{equation}
	\begin{aligned}
	\hat{\mathbf{Z}} = \mathop{\arg \min}_{\mathbf{Z}} \Vert \mathbf{Z} \Vert_{0,2},	\text{ s.t. } \Vert \hat{\mathbf{Y}} - \mathbf{A} \mathbf{Z} \Vert_{F}^{2} \le \epsilon,
	\label{blindrec}
	\end{aligned}
	\end{equation}
with $\hat{\mathbf{Y}}$ being the constructed sub-samples and $\epsilon$ being the pre-specified bound for noise. It is noted that the constructed sub-samples can be either a matrix or a vector, which depends on the output of the CSL algorithm (corresponding to the vCSLACC algorithm and the mCSLSACC algorithm in Section \ref{pCSL}). Here, we use the matrix form, namely $\hat{\mathbf{Y}}$, without loss of generality. The matrix form denoting automatically reduces to the vector form if the output of the mCSLSACC algorithm is employed and the matrix $\mathbf{Z}$ correspondingly changes to the vector $\mathbf{z}$. To solve \eqref{blindrec}, we choose a greedy algorithm for the low computational complexity and simultaneous orthogonal matching pursuit (SOMP) \cite{Tropp2005Simultaneous} is properly used to achieve a feasible solution to the blind recovery problem in \eqref{blindrec}. We can alternatively integrate the constructed SCMs in \eqref{eq:33} with the framework in \cite{Ariananda2012Compressive,Cohen2014Sub,Romero2016Compressive} to realize the PSD recovery. However, as the emphasis of this paper lies in the influence analysis of the antenna and channel to the proposed CSL algorithms, we choose to adopt the spectrum recovery in \eqref{blindrec} instead of the PSD recovery.

After the recovery stage, we use the energy detection method to make the final decision based on the recovered spectrum supports $\mathbb{S}$ and Fourier coefficients $\mathbf{Z}$. Specifically, for each $i \in \mathbb{S}$, the test {statistics} is given by $\mathcal{T}_{i} = \frac{1}{M}\sum_{m=1}^{M} |\mathbf{z}_{m}(i)|^{2}$. We define the probability of detection $p_{d}^{i}$ and the probability of false alarm $p_{f}^{i}$ for each spectrum support $i \in \mathbb{S}$, where $p_{d}^{i} = \mathbb{P}(\mathcal{T}_{i} > \gamma | \mathcal{H}_{1}^{i})$ and $p_{f}^{i} = \mathbb{P}(\mathcal{T}_{i} > \gamma | \mathcal{H}_{0}^{i})$ with $\mathcal{H}_{1}^{i}$ and $\mathcal{H}_{0}^{i}$ indicating the presence and absence of PUs on support $i$, $\gamma$ being a pre-defined non-negative threshold and $\mathbb{P}(\cdot)$ denoting the probability. Furthermore, let $\bar{\mathbb{S}}$ denote the complementary set of $\mathbb{S}$ such that $\mathbb{S} \cup \bar{\mathbb{S}}$ exactly consists of the whole index set of disjoint narrowbands. Then, the probabilities of detection and false alarm in terms of the spectrum support set $\mathbb{S}$ are defined as the averages of $p_{d}^{i}$ and $p_{f}^{i}$ over $\mathbb{S}$ and $\bar{\mathbb{S}}$, respectively, i.e.,
\begin{equation}
\begin{aligned}
	P_d = \frac{1}{|\mathbb{S}|} \sum\nolimits_{i \in \mathbb{S} } p_{d}^{i}
	\text{ and }
	P_f  = \frac{1}{|\bar{\mathbb{S}}|} \sum\nolimits_{i \in \bar{\mathbb{S}} } p_{f}^{i}.
	\label{PDPF}
\end{aligned}
\end{equation}
The entire process of the CSL based WBSS can be seen from the digital domain {in} Fig. \ref{TotalSystem}.

%===================================================================%
%--------------Compressive Subspace Learning algorithms----------------%
%===================================================================%

\section{The Proposed CSL Algorithms}
\label{SectionIII}

Let $\mathscr{A}$ be an operation set, whose elements are to realize specified arrangements of sub-samples. Thus, $\mathscr{A}(\mathbf{Y})$ represents the reorganized sub-samples after the operation $\mathscr{A}$. Given two sub-sample matrices $\mathbf{Y}_{1}$ and $\mathbf{Y}_{2}$ obtained from two antenna sub-arrays, the SCM in terms of $\mathscr{A}$ is further given by $\mathbf{R}(\mathbf{Y}_{1},\mathbf{Y}_{2},\mathscr{A}) \triangleq \mathscr{E}\{ \mathscr{A}(\mathbf{Y}_{1}) \mathscr{A}^{H} (\mathbf{Y}_{2}) \}$. This SCM is generally low-rank with several dominant singular values and can be approximated by a compact subspace. We first denote the number of dominant singular values as $s$ which is unknown in advance. The main aim of CSL is to find such a compact subspace by solving the following problem 
	\begin{equation}
	\begin{aligned}
		\tilde{\mathbf{U}} = \arg \min_{\mathbf{U} \in \mathbb{C}^{P \times s}}& \| (\mathbf{I} - \mathbf{U} \mathbf{U}^{H}) \mathbf{R}(\mathbf{Y}_{1},\mathbf{Y}_{2},\mathscr{A}) \|_{F}^{2} \\
		s.t. \quad&  \mathbf{U}^{H} \mathbf{U} = \mathbf{I}.
	\label{optimization}
	\end{aligned}
	\end{equation}
We then directly obtain a simple solution $\tilde{\mathbf{U}}$ by performing the singular value decomposition (SVD) to $\mathbf{R}(\mathbf{Y}_{1},\mathbf{Y}_{2},\mathscr{A})$ and then picking columns of the unitary matrix corresponding to the $s$ largest singular values in an adaptive manner. For different CSL algorithms, the $\mathbf{R}(\mathbf{Y}_{1},\mathbf{Y}_{2},\mathscr{A})$s are different, leading to the fact that different CSL algorithms have different values of $s$. The proposed CSL algorithms in Section~\ref{pCSL} also follow such a fact.

Before we formally propose new CSL algorithms, the traditional CSL algorithms\cite{Ioushua2017CaSCADE, Zhang2018Distributed,Ma2016Reliable,Zeng2009Eigenvalue, Li2016Maximum, Liu2017Optimal, Bouallegue2018Blind, Jin2019Spectrum} neither considering spatially correlated MIMO channel nor using antenna cross-correlations are first described in a unified model in order to make the following analyses clear. For spatially uncorrelated MIMO channel, the receiving correlation matrix of Kronecker model \eqref{KroneckerModel} reduces to an identity matrix, resulting in $\mathbf{G} = \mathbf{G}_{w} \mathbf{P}^{\frac{1}{2}}$. Without using antenna cross-correlations, in other words, traditional CSL algorithms only exploit antenna auto-correlations ($\mathbf{Y}_{1}$ and $\mathbf{Y}_{2}$ obtained from the same antenna sub-arrays, i.e., $\mathbf{Y}_{1} = \mathbf{Y}_{2}$). We can reformulate the SCMs of traditional CSL algorithms in a unified form as $\mathbf{R}(\mathbf{Y},\mathscr{A})$ ($\mathbf{Y} \triangleq \mathbf{Y}_{1} = \mathbf{Y}_{2}$). The compact subspace of $\mathbf{R}(\mathbf{Y},\mathscr{A})$ can be achieved via using the SVD to $\mathbf{R}(\mathbf{Y},\mathscr{A})$.

\begin{figure*}[t!]
	\centering
	\includegraphics[height=3.6cm, width=18.0cm]{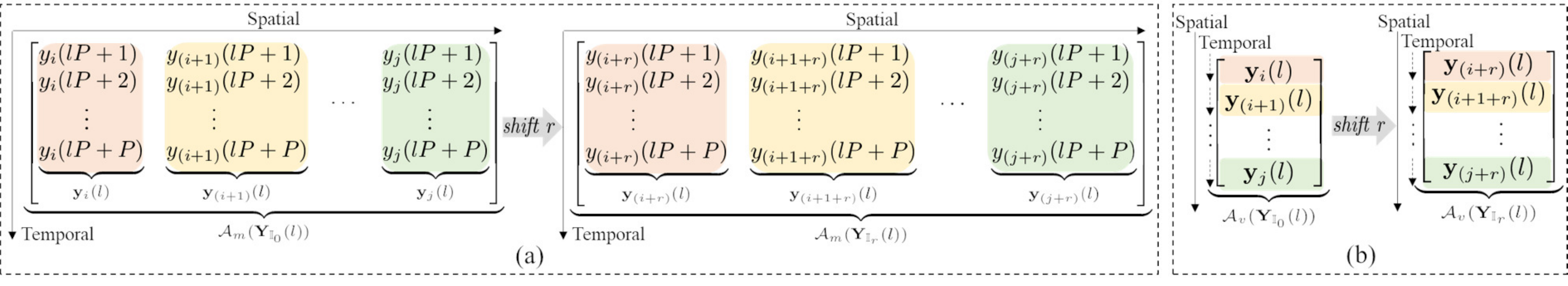}
	\vspace{-0.8em}
	\caption{The considered sub-sample arrangements: (a) matrix form, (b) vector form.}
	\label{MFVF}
\end{figure*}

\subsection{The Proposed mCSLSACC and vCSLACC Algorithms}
\label{pCSL}
We propose to exploit antenna cross-correlations and {consider} two different sub-sample operations ($\mathcal{A}_{m}$, $\mathcal{A}_{v} \in \mathscr{A}$) to realize the matrix and vector arrangements, respectively. Based on this idea, two CSL algorithms (i.e., mCSLSACC and vCSLACC) are proposed. We assume that in a sensing period each sub-sample data stream of an antenna is divided into $L$ segments with each segment $P$ sub-samples. In addition, we introduce an integer parameter $r (0 \le r \le M-1)$ as the shift factor of antenna array to indicate the index shifting of a sub-array compared with another sub-array. Among all $M$ antennas, we select a consecutive sub-array indexed by $\mathbb{I}_{0} \triangleq \{i, i+1, \cdots, j\} (1 \le i \le j \le M-r)$ to generate the basic sub-sample group. Then another consecutive sub-array indexed by $\mathbb{I}_{r} \triangleq \{i+r, i+1+r, \cdots, j+r\}$ produces the shifted sub-sample group. When $r = 0$, it is the antenna auto-correlation case, which is also considered for completeness. Without loss of generality, we consider the $l$th $(l \in [1,L])$ segment and {give the following equations without including index $l$}. The specific form by $\mathcal{A}_{m}$ shown in Fig. \ref{MFVF}(a) is given by
	\begin{flalign}
	\label{MatrixArrange1}
	\mathcal{A}_{m}(\mathbf{Y}_{\mathbb{I}_{0}}) &= \mathbf{Y}_{\mathbb{I}_{0}} = [\mathbf{Y}]_{:, i:j} \\
	\label{MatrixArrange2}
	\mathcal{A}_{m}(\mathbf{Y}_{\mathbb{I}_{r}}) &= \mathbf{Y}_{\mathbb{I}_{r}} = [\mathbf{Y}]_{:, (i+r):(j+r)}, 
	\end{flalign}
where $\mathcal{A}_{m}(\mathbf{Y}_{\mathbb{I}_{0}}), \mathcal{A}_{m}(\mathbf{Y}_{\mathbb{I}_{r}}) \in \mathbb{C}^{P \times (j-i+1)}$. For each $r$, we define $\mathbf{R}(\mathbf{Y},\mathcal{A}_{m},\mathbb{I}_{0},r) \triangleq \mathscr{E}\{\mathcal{A}_{m}(\mathbf{Y}_{\mathbb{I}_{0}})$ $\mathcal{A}_{m}^{H}(\mathbf{Y}_{\mathbb{I}_{r}})\}$ as an individual SCM. By collecting all individual cross-correlations to form a combined SCM $\mathbf{R}_{Y_c} \triangleq \sum_{r = 1}^{i-1} \mathbf{R}(\mathbf{Y},\mathcal{A}_{m},\mathbb{I}_{0},-r) + \sum_{r = 1}^{M-j} \mathbf{R}(\mathbf{Y},\mathcal{A}_{m},\mathbb{I}_{0},r)$ and learning the signal subspace from $\mathbf{R}_{Y_c}$, we derive the mCSLSACC. On the other hand, the vCSLACC is developed by learning signal subspace from {a} defined SCM $\mathbf{R}(\mathbf{Y},\mathcal{A}_{v},\mathbb{I}_{0},r) \triangleq \mathscr{E}\{\mathcal{A}_{v}(\mathbf{Y}_{\mathbb{I}_{0}}) \mathcal{A}_{v}^{H}(\mathbf{Y}_{\mathbb{I}_{r}})\}$, in which the specific form by $\mathcal{A}_{v}$ shown in Fig. \ref{MFVF}(b) is given by
	\begin{flalign}
	\label{VectorArrange1}
	\mathcal{A}_{v}(\mathbf{Y}_{\mathbb{I}_{0}}) &= \mathbf{y}_{\mathbb{I}_{0}} = vec\left([\mathbf{Y}]_{:, i:j}\right) \\
	\label{VectorArrange2}
	\mathcal{A}_{v}(\mathbf{Y}_{\mathbb{I}_{r}}) &= \mathbf{y}_{\mathbb{I}_{r}} = vec\left([\mathbf{Y}]_{:, (i+r):(j+r)}\right),
	\end{flalign}
where $\mathcal{A}_{v}(\mathbf{Y}_{\mathbb{I}_{0}}), \mathcal{A}_{v}(\mathbf{Y}_{\mathbb{I}_{r}}) \in \mathbb{C}^{P(j-i+1) \times 1}$. The operation by $\mathcal{A}_{v}$ stacks sub-samples of each antenna into a vector in an increasing order from the $i$th ($i+r$th) to the $j$th ($j+r$th) antenna.

As stated in the solving process of problem \eqref{optimization}, the signal subspace is obtained via an SVD, which is given by
	\begin{equation}
	\mathbf{R}(\mathbf{Y},\mathcal{A},\mathbb{I}_{0},r) = [\mathbf{U}_{s}, \mathbf{U}_{n}] 
	blkdiag\{ \mathbf{D}_{s}, \mathbf{D}_{n} \}
	[\mathbf{U}_{s}, \mathbf{U}_{n}]^{H},
	\label{eq:32}
	\end{equation}
where $\mathcal{A} \in \{\mathcal{A}_{m}, \mathcal{A}_{v}\}$; $\mathbf{U}_{s}$ and $\mathbf{U}_{n}$ are matrices constituted by {eigenvectors corresponding to} signal and noise, respectively; $\mathbf{D}_{s} \triangleq diag\{\lambda_{1}, \lambda_{2}, \cdots, \lambda_{s}\}$ and $\mathbf{D}_{n} \triangleq diag\{\lambda_{s+1}, \lambda_{s+2}, \cdots, \lambda_{s+n}\}$ are diagonal matrices with singular values {indicating signal and noise} as their diagonal elements, respectively. The signal subspace $\mathbf{U}_{s} = [\mathbf{u}_{1}, \mathbf{u}_{2}, \cdots, \mathbf{u}_{s}]$ is determined by finding the $s$ largest singular values and the remaining $n$ singular values are used for noise subspace extraction, {thereby resulting in a signal subspace projection $\mathbf{P}_{s} = \mathbf{U}_{s} \mathbf{U}_{s}^{H}$}. We can construct the SCM and sub-samples via the obtained signal subspace as
\begin{flalign}
    \label{eq:33}
    \hat{\mathbf{R}} = \mathbf{P}_{s} \mathbf{R} \quad \text{and} \quad
    %\label{eq:34}
    \hat{\mathbf{y}}_{c} = \mathbf{U}_{s} \boldsymbol{\lambda}_{s},
\end{flalign}
where $\boldsymbol{\lambda}_{s} \triangleq [\lambda_{1}, \lambda_{2}, \cdots, \lambda_{s}]^{T}$. In practice, we estimate the SCM from finite sub-samples by averaging $L$ cross-correlations. The proposed algorithms are listed in Algorithm \ref{alg1} and Algorithm \ref{alg2}, respectively.

\begin{algorithm}[!t]
	\caption{mCSLSACC algorithm}
	\label{alg1}
	\footnotesize
	\KwIn{$i$, $j$, $\{\mathbf{y}_{m}(l)\}_{m = 1 \cdots M, l = 1 \cdots L}$}
	\KwOut{$\hat{\mathbf{y}}_{c}$}
		\vspace{.05cm}
		$\mathbb{I}_{0} \gets \{i, i+1, \cdots, j\}$ \tcp*{Get the index set of basic antenna sub-array}
		$\mathbb{I}_{r} \gets \{i+r, i+1+r, \cdots, j+r\}$ \tcp*{Get the index set of shifted antenna sub-array}
		$\mathbf{Y}(l) \gets \left[ \mathbf{y}_{1}(l), \mathbf{y}_{2}(l) \cdots, \mathbf{y}_{M}(l) \right]$ \tcp*{Stack sub-samples into a matrix} 
		$\hat{\mathbf{R}}(\mathbf{Y},\mathcal{A}_{m},\mathbb{I}_{0},r) \gets \frac{1}{L} \sum_{l=1}^{L} \mathcal{A}_{m}(\mathbf{Y}_{\mathbb{I}_{0}}(l)) \mathcal{A}_{m}^{H}(\mathbf{Y}_{\mathbb{I}_{r}}(l))$ \tcp*{Estimate individual SCMs over $L$ segments}
		$\hat{\mathbf{R}}_{Y_c} \gets \sum\nolimits_{r = 1}^{i-1} \hat{\mathbf{R}}(\mathbf{Y},\mathcal{A}_{m},\mathbb{I}_{0},-r) + \sum\nolimits_{r = 1}^{M-j} \hat{\mathbf{R}}(\mathbf{Y},\mathcal{A}_{m},\mathbb{I}_{0},r)$ \tcp*{Get a combined SCM using individual SCMs}
		$[\mathbf{U}_{s}, \mathbf{U}_{n}] blkdiag\{ \mathbf{D}_{s}, \mathbf{D}_{n} \} [\mathbf{U}_{s}, \mathbf{U}_{n}]^{H} \gets \text{SVD}(\hat{\mathbf{R}}_{Y_c})$ \tcp*{Perform SVD to the combined SCM $\hat{\mathbf{R}}_{Y_c}$}
		$\hat{\mathbf{y}}_{c} \gets \mathbf{U}_{s} \cdot [\lambda_{1}, \lambda_{2}, \cdots, \lambda_{s}]^{T}$ \tcp*{Construct approximately clean sub-samples from the signal subspace}
		\textbf{return} $\hat{\mathbf{y}}_{c}$ 
\end{algorithm}

It is worth noting that to learn subspaces from {$\mathscr{E}\{ \mathcal{A}(\mathbf{Y}_{\mathbb{I}_{0}}) \mathcal{A}^{H}(\mathbf{Y}_{\mathbb{I}_{r}}) \}$} is equivalent to learn the subspaces from $\mathscr{E}\{ \mathcal{A}(\mathbf{Y}_{\mathbb{I}_{r}}) \mathcal{A}^{H}(\mathbf{Y}_{\mathbb{I}_{0}}) \}$ {for} $\mathcal{A} \in \{\mathcal{A}_{m}, \mathcal{A}_{v}\}$. Therefore, we only take {the former} into consideration.
In addition, we summarize the main differences of two proposed algorithms as follows:
\begin{itemize}
    \item For all different values of $r$, the mCSLSACC corresponds to the case of only one input-output vector pair for its decomposed SCM being a combined result over $r$, while the vCSLACC includes $r$ input-output vector pairs. 
    In addition, the number of dominant singular values (i.e., $s$) are different for both algorithms.
    
    \item 
    The basic SCM used to constitute {$\mathbf{R}_{Y_c}$ in the} mCSLSACC is a sum of individual $\mathscr{E}\{ \mathbf{y}_{m} \mathbf{y}_{m+r}^{H} \}$ $(i \le m \le j)$ with equal weights over $j-i+1$ antennas. Therefore, the signal subspace is obtained as a superposition result and its bases expands on the time dimension. {Nevertheless}, the SCM in the vCSLACC is constituted by sub-matrices $\mathscr{E}\{ \mathbf{y}_{m_{1}} \mathbf{y}_{m_{2}}^{H} \} (i \le m_1 \le j, i+r \le m_2 \le j+r)$, leading to a spatial-temporal joint decomposition.
\end{itemize}

\begin{algorithm}[!t]
	\caption{vCSLACC algorithm}
	\label{alg2}
	\footnotesize
	\KwIn{$i$, $j$, $r_{0}$, $\{\mathbf{y}_{m}(l)\}_{m = 1 \cdots M, l = 1 \cdots L}$}
	\vspace{-.03cm}
	\KwOut{$\hat{\mathbf{Y}}$} 
	\vspace{.01cm}
	$\mathbb{I}_{0} \gets \{i, i+1, \cdots, j\}$, \tcp*{Get the index set of basic antenna sub-array}
	$r \gets r_{0}$; 
	$\mathbb{I}_{r} \gets \{i+r, i+1+r, \cdots, j+r\}$ \tcp*{Get the index set of shifted antenna sub-array}
	$\mathbf{Y}(l) \gets \left[ \mathbf{y}_{1}(l), \mathbf{y}_{2}(l) \cdots, \mathbf{y}_{M}(l) \right]$ \tcp*{Stack the sub-samples into a matrix}
	$\hat{\mathbf{R}}(\mathbf{Y},\mathcal{A}_{v},\mathbb{I}_{0},r) \gets \frac{1}{L} \sum_{l=1}^{L} \mathcal{A}_{v}(\mathbf{Y}_{\mathbb{I}_{0}}(l)) \mathcal{A}_{v}^{H}(\mathbf{Y}_{\mathbb{I}_{r}}(l))$ \tcp*{Estimate the SCMs over $L$ segments}
	$[\mathbf{U}_{s}, \mathbf{U}_{n}] blkdiag\{ \mathbf{D}_{s}, \mathbf{D}_{n} \} [\mathbf{U}_{s}, \mathbf{U}_{n}]^{H} \gets \text{SVD}(\hat{\mathbf{R}}(\mathbf{Y},\mathcal{A}_{v},\mathbb{I}_{0},r))$ \tcp*{Perform SVD to SCM $\hat{\mathbf{R}}(\mathbf{Y},\mathcal{A}_{v},\mathbb{I}_{0},r)$}
	$\hat{\mathbf{y}}_{c} \gets \mathbf{U}_{s} \cdot [\lambda_{1}, \lambda_{2}, \cdots, \lambda_{s}]^{T}$ \tcp*{Construct approximately clean sub-samples from the signal subspace}
	$\hat{\mathbf{Y}} \gets reshape(\hat{\mathbf{y}}_{c}, [P, j-i+1])$ \tcp*{Reshape the sub-sample vector into a matrix}
	\textbf{return} $\hat{\mathbf{Y}}$
\end{algorithm}

\section{Matrix Form CSL with The Sum of Antenna Cross-correlations}
\label{SectionIV}
In this section, we establish the SCM relations between the proposed mCSLSACC with traditional CSL algorithms and then analyze the space diversity and shift factor of the antenna array through the SCM relations. Since the SCM in mCSLSACC is formed by collecting the sum of cross-correlations over $r$, we investigate the mCSLSACC via its basic part, i.e., for a fixed $r$, we investigate the subspace learning based on $\mathbf{R}(\mathbf{Y},\mathcal{A}_{m},\mathbb{I}_{0},r) = \mathscr{E}\{\mathcal{A}_{m}(\mathbf{Y}_{\mathbb{I}_{0}}) \mathcal{A}_{m}^{H}(\mathbf{Y}_{\mathbb{I}_{r}})\}$, which we refer to as the mCSLACC (i.e., matrix form CSL with antenna cross-correlations). The mCSLSACC will be presented based on the results of the mCSLACC.

\subsection{Statistical Covariance Matrix}
\label{FOURA}
To {establish} the SCM of the mCSLACC, we arrange the corresponding Nyquist samples $\mathbf{s}_{k} = [s_{k}(lQ+1), s_{k}(lQ+2), \cdots, s_{k}(lQ+Q)]^{T}$ and $\mathbf{c}_{m} = [c_{m}(lQ+1), c_{m}(lQ+2), \cdots, c_{m}(lQ+Q)]^{T}, 1 \le m \le M, c \in \{ x, n \}$ into matrices as 
	\begin{equation}
	\begin{aligned}
	\mathbf{S} &= [\mathbf{s}_{1}, \mathbf{s}_{2}, \cdots, \mathbf{s}_{K}],\\
	\mathbf{C} &= [\mathbf{c}_{1}, \mathbf{c}_{2}, \cdots, \mathbf{c}_{M}],
	\label{MatrixForm}
	\end{aligned}
	\end{equation}
where $\mathbf{c} \in \{ \mathbf{x}, \mathbf{n} \}$, $\mathbf{C} \in \{\mathbf{X}, \mathbf{N}\}$.
Defining $\mathbf{C}_{\mathbb{I}_{0}} \triangleq [\mathbf{C}]_{:, i:j}$, $\mathbf{C}_{\mathbb{I}_{r}} \triangleq [\mathbf{C}]_{:, (i+r):(j+r)}$ and following \eqref{eqMAGMCmat}, we get 
	\begin{flalign}
    \label{MatRelationship_1}
    \mathbf{Y}_{\mathbb{I}_{0}} &= \mathbf{\Omega} \mathbf{S} \mathbf{G}_{\mathbb{I}_{0}}^{T} + \mathbf{N}_{\mathbb{I}_{0}} \\
    \label{MatRelationship_2}
    \mathbf{Y}_{\mathbb{I}_{r}} &= \mathbf{\Omega} \mathbf{S} \mathbf{G}_{\mathbb{I}_{r}}^{T} + \mathbf{N}_{\mathbb{I}_{r}},
    \end{flalign}
where $\mathbf{G}_{\mathbb{I}_{0}}^{T} \triangleq [\mathbf{G}^{T}]_{:,i:j}$ and $\mathbf{G}_{\mathbb{I}_{r}}^{T} \triangleq [\mathbf{G}^{T}]_{:,(i+r):(j+r)}$.
Then, $\mathbf{R}(\mathbf{Y},\mathcal{A}_{m},\mathbb{I}_{0},r)$ can be further given by
	\begin{flalign}
	\nonumber
	\mathbf{R}(\mathbf{Y},\mathcal{A}_{m},\mathbb{I}_{0},r) 
	&= \mathbf{\Omega} \sum_{k=1}^{K} [\mathbf{g}_{k}^{T}]_{\mathbb{I}_{0}} [\mathbf{g}_{k}^{*}]_{\mathbb{I}_{r}} \mathbf{R}_{s_{k}} \mathbf{\Omega}^{H} \\
	\label{MatYS}
	&+ \mathbf{R}(\mathbf{N},\mathcal{A}_{m},\mathbb{I}_{0},r),
	\end{flalign}
where $\mathbf{g}_{k}$ is {denoted as} the $k$th column of $\mathbf{G}$ with $[\mathbf{g}_{k}^{T}]_{\mathbb{I}_{0}} \triangleq [\mathbf{g}_{k}^{T}]_{i:j}$ and $[\mathbf{g}_{k}^{T}]_{\mathbb{I}_{r}} \triangleq [\mathbf{g}_{k}^{T}]_{(i+r):(j+r)}$; $\mathbf{R}_{s_{k}} \triangleq \mathscr{E}\left\{ \mathbf{s}_{k}\mathbf{s}_{k}^{H} \right\}$ and $\mathbf{R}(\mathbf{N},\mathcal{A}_{m},\mathbb{I}_{0},r) \triangleq \mathscr{E}\{\mathcal{A}_{m}(\mathbf{N}_{\mathbb{I}_{0}}) \mathcal{A}_{m}^{H}(\mathbf{N}_{\mathbb{I}_{r}})\}$. \eqref{MatYS} is derived by following the fact that $\mathscr{E}\left\{ \mathbf{s}_{k_{1}} \mathbf{s}_{k_{2}}^{H} \right\} = \mathbf{0}$ when $k_{1} \ne k_{2}$. For the noise term $\mathbf{R}(\mathbf{N},\mathcal{A}_{m},\mathbb{I}_{0},r)$, discussions are listed as follows:
\begin{itemize}
    \item 
    {At} $r = 0$, the noise term {comes from an auto-correlation of $\mathbf{N}_{\mathbb{I}_{0}}$, thereby leading to the {only use} of antenna auto-correlations.} It, thus, becomes a diagonal matrix $\sigma_{N_{\mathbb{I}_{0}}}^{2} \mathbf{I}_{P}$ with $\sigma_{N_{\mathbb{I}_{0}}}^{2} \triangleq \sum_{m \in \mathbb{I}_{0}} \alpha_{m} \sigma_{m}^{2}$ and $\alpha_{m}$ being the noise folding factor. {This condition leads to a noise {corruption to} $\mathbf{R}(\mathbf{Y},\mathcal{A}_{m},\mathbb{I}_{0},r)$.}
    \item
    When $r \neq 0$, the noise term is obtained based on the cross-correlation between $\mathbf{N}_{\mathbb{I}_{0}}$ and $\mathbf{N}_{\mathbb{I}_{r}}$, allowing the exploitation of antenna cross-correlations without {being faced with} the noise corruption at $r = 0$. This reveals that $\mathbf{R}(\mathbf{Y},\mathcal{A}_{m},\mathbb{I}_{0},r)$ is not interfered by noise for $r \neq 0$.
\end{itemize}
{Throughout the above discussions, we can exploit the second condition to improve the signal subspace learning in the low SNR scenario.}

To establish the SCM relation between the basic mCSLACC and traditional CSL algorithms for further analysis, both SCMs in statistical sense w.r.t. the wireless fading channel are investigated. Following the Kronecker model in \eqref{KroneckerModel}, we {obtain $\mathbf{g}_{k} = \sigma_{k} \mathbf{Q}^{\frac{1}{2}} \bar{\mathbf{g}}_{k}$}, where {$\sigma_{k}$} is the $k$th diagonal element of $\mathbf{P}^{\frac{1}{2}}$ and $\bar{\mathbf{g}}_{k}$ is the $k$th column of matrix $\mathbf{G}_{w}$ whose elements follow the Rayleigh distribution with zero mean and unit variance. Let $\mathbf{q}_{m} (1 \le m \le M)$ be columns of $\mathbf{Q}^{\frac{1}{2}}$. Then we get 
	\begin{flalign}
	\nonumber
	\mathscr{E}\left\{ [\mathbf{g}_{k}^{T}]_{\mathbb{I}_{0}} [\mathbf{g}_{k}^{*}]_{\mathbb{I}_{r}} \right\}
		&= \sigma_{k}^{2} \sum_{m=1}^{M} \left[\mathbf{q}_{m}^{T} \right]_{\mathbb{I}_{0}} \left[\mathbf{q}_{m}^{*} \right]_{\mathbb{I}_{r}} \\
		\label{MatGain}
		&= \sigma_{k}^{2} \sum_{u \in \mathbb{I}_{0}} \mathbf{q}_{u}^{H} \mathbf{q}_{u+r},
	\end{flalign}
{where} $\mathscr{E}\left\{\bar{g}_{m_1k} \bar{g}_{m_2k}^{*}\right\} = 0, \forall m_1 \ne m_2$ and $\mathscr{E}\left\{ \left| \bar{g}_{mk} \right|^{2} \right\} = 1, \forall m$ are used for the derivation. The last equality {in \eqref{MatGain} follows} the Hermitian structure of $\mathbf{Q}^{\frac{1}{2}}$.

Taking expectation to both sides of \eqref{MatYS} and using {\eqref{MatGain}}, we obtain
    \begin{flalign}
	\label{MatYS2}
	\bar{\mathbf{R}}(\mathbf{Y},\mathcal{A}_{m},\mathbb{I}_{0},r) 
	= \sum_{u \in \mathbb{I}_{0}} {\mathbf{q}_{u}^{H} \mathbf{q}_{u+r}}  \bar{\mathbf{R}}_{sa} + \mathbf{R}(\mathbf{N},\mathcal{A}_{m},\mathbb{I}_{0},r),
	\end{flalign}
where $\bar{\mathbf{R}}(\mathbf{Y},\mathcal{A}_{m},\mathbb{I}_{0},r) \triangleq \mathscr{E}_{\mathbf{g}_{k}} \left\{ \mathbf{R}(\mathbf{Y},\mathcal{A}_{m},\mathbb{I}_{0},r) \right\}$ represents the SCM in statistical sense w.r.t. the wireless fading channel and $\bar{\mathbf{R}}_{sa} \triangleq \mathbf{\Omega} \sum_{k=1}^{K} {\sigma_{k}^{2}} \mathbf{R}_{s_{k}} \mathbf{\Omega}^{H}$. The matrix $\bar{\mathbf{R}}_{sa}$ coincides with the SCM in statistical sense of traditional CSL algorithms considering the MIMO channel without spatial correlations ({i.e., by setting $\mathbf{g}_{k} = {\sigma_{k}} \mathbf{I}_{M} \bar{\mathbf{g}}_{k}$}).
{It is noted from \eqref{MatYS2} that the mCSLACC introducs $G \triangleq \sum\nolimits_{u \in \mathbb{I}_{0}} {\mathbf{q}_{u}^{H} \mathbf{q}_{u+r}}$ over traditional CSL algorithms in statistical sense in terms of SCM.}

Based on the derived relation in \eqref{MatYS2}, we {next} show how the spatially correlated MIMO channel and the shift factor of antenna sub-arrays influence the performance of {the} mCSLACC by comparing it with traditional CSL algorithms. The comparison is based on a relation between singular values of two SCMs. Such a singular value relation helps us understand the amplification introduced by {the} mCSLACC, where the large amplification on singular values allows an accurate discrimination from noise than traditional CSL algorithms.

\subsection{Analysis on Space Diversity and Shift Factor}
We define singular value matrices $\mathbf{D}(\mathbf{Y},\mathcal{A}_{m},\mathbb{I}_{0},r) \triangleq sv \left( \bar{\mathbf{R}}(\mathbf{Y},\mathcal{A}_{m},\mathbb{I}_{0},r) \right)$ and $\mathbf{D}_{sa} \triangleq sv\left(\bar{\mathbf{R}}_{sa}\right)$, {allowing the derivation of {the} singular value relation {in} \eqref{MatYS2} as}
	\begin{flalign}
	\label{MatEigRelation}
	\mathbf{D}(\mathbf{Y},\mathcal{A}_{m},\mathbb{I}_{0},r) 
	&=\left\{ \begin{aligned}
	&\sum\limits_{u \in \mathbb{I}_{0}} {\mathbf{q}_{u}^{H} \mathbf{q}_{u}} \mathbf{D}_{sa} + \sigma_{N_{\mathbb{I}_{0}}}^{2} \mathbf{I}_{P},&  &r = 0 \\
	&\left| \sum\limits_{u \in \mathbb{I}_{0}} {\mathbf{q}_{u}^{H} \mathbf{q}_{u+r}} \right| \mathbf{D}_{sa},& &r \ne 0. \\
	\end{aligned}  \right.
	\end{flalign}
We note that the singular values in $\mathbf{D}(\mathbf{Y},\mathcal{A}_{m},\mathbb{I}_{0},r)$ are amplified results of the singular values in $\mathbf{D}_{sa}$ by an amplification factor $|G| \triangleq \left| \sum\nolimits_{u \in \mathbb{I}_{0}} {\mathbf{q}_{u}^{H} \mathbf{q}_{u+r}} \right|$.

In order to provide more insights on the derived amplification factor, the exponential correlation model is used to describe the receiving correlation matrix, whose specific form is given by
	\begin{equation}
	\begin{aligned}
	\mathbf{Q}
	=\begin{bmatrix}
	1 	 	   &\rho     &\cdots 	&\rho^{M-1} \\
	\rho^{*}  	 	   &1  	       &\cdots     &\rho^{M-2} \\
	\vdots 	   &\vdots        &\ddots 	&\vdots \\
	(\rho^{*})^{M-1}   &(\rho^{*})^{M-2}  &\cdots 	&1 \\
	\end{bmatrix},
	\label{ECM}
	\end{aligned}
	\end{equation}
where $\rho$ is a coefficient (known as correlation coefficient) characterizing the correlation between adjacent antennas. This model {shows that the correlations are the same for equally spaced antennas and are exponentially decayed for increasingly spaced antennas.} {{With such an} exponential correlation model} and the equation $\mathbf{Q}^{\frac{1}{2}} (\mathbf{Q}^{\frac{1}{2}})^{H} = \mathbf{Q}$, we obtain the correlation between the $m_1$th and $m_2$th columns of $\mathbf{Q}^{\frac{1}{2}}$ as
	\begin{flalign}
	\label{columncorr}
	\mathbf{q}_{m_1}^{H} \mathbf{q}_{m_2} 
	= \left\{ \begin{aligned}
	&\rho^{m_1-m_2},& 			&m_1<m_2, \\
	&1,&  						&m_1=m_2, \\
	&(\rho^{*})^{m_2-m_1},& 	&m_1>m_2. \\
	\end{aligned}  \right.
	\end{flalign}
{Accordingly}, the amplification factor $|G|$ is obtained based on \eqref{columncorr} as
	\begin{flalign}
	\label{svGain}
	\left| G \right| 
	=\left| \sum\limits_{u \in \mathbb{I}_{0}} {\mathbf{q}_{u}^{H} \mathbf{q}_{u+r}} \right|
	= (j-i+1)|\rho|^{r}.
	\end{flalign}
As in \eqref{svGain}, the $|G|$ reduces to $j-i+1$ {at} $r=0$, which is equivalent to traditional CSL algorithms considering spatially uncorrelated MIMO channel. Dividing the noise term $\sigma_{N_{\mathbb{I}_{0}}}^{2} \mathbf{I}_{P}$ by $j-i+1$, the mCSLACC can be considered as an antenna averaging scheme, which makes the mCSLACC robust against noise. When $r \ne 0$, the $|G|$ reduces by the correlation coefficient $|\rho|$ with an increasing $r$, which reveals that the amplification effect by antenna cross-correlations becomes weak with the increasing of antenna distances. However, the superiority of mCSLACC at $r \ne 0$ over that at $r = 0$ is attributed to the lack of noise corruption. Moreover, in the case of $r \ne 0$, a large correlation coefficient indeed promotes an improved amplification effect, while in the case of $r = 0$, the fact $|\rho|^{0} = 1$ shows that $|G|$ does not change with varying correlation coefficients. As expected, both cases are influenced by the number of deployed antennas $j-i+1$, where both of their performances improve with the increasing $j-i+1$.

\subsection{Space Diversity and Shift Factor for The mCSLSACC}
As previously discussed, we can tell from \eqref{MatEigRelation} that, for each $r \neq 0$, the SCM of mCSLACC enjoys noise-free property and introduces a $|G|$ over the singular values of SCM in traditional CSL algorithms. Since $|G|$ exponentially decayed, taking advantage of all noise-free amplifications is necessary. To this end, we propose the mCSLSACC. By summing up individual SCMs $\mathbf{R}(\mathbf{Y},\mathcal{A}_{m},\mathbb{I}_{0},r)$ $(1 \le r \le M-j)$ over $r$ for a given sub-array indexed by $\mathbb{I}_{0}$, we obtain the statistical relation based on \eqref{MatYS2} between the combined SCM $\mathbf{R}_{Y_c}$ and the SCM of traditional CSL algorithms as
    \begin{flalign}
	\label{MatYSCombined}
	\bar{\mathbf{R}}_{Y_c}
	= \sum_{r = 1}^{i-1} \sum_{u \in \mathbb{I}_{0}} {\mathbf{q}_{u}^{H} \mathbf{q}_{u-r}}  \bar{\mathbf{R}}_{sa} + \sum_{r = 1}^{M-j} \sum_{u \in \mathbb{I}_{0}} {\mathbf{q}_{u}^{H} \mathbf{q}_{u+r}}  \bar{\mathbf{R}}_{sa}.
	\end{flalign}
Then, {based on \eqref{MatYSCombined},} the singular value relation is given by
	\begin{flalign}
	\nonumber
	&\mathbf{D}_{Y_{c}}
	= \sum_{r=1}^{i-1} \mathbf{D}(\mathbf{Y},\mathcal{A}_{m},\mathbb{I}_{0},-r) + \sum_{r=1}^{M-j} \mathbf{D}(\mathbf{Y},\mathcal{A}_{m},\mathbb{I}_{0},r) \\
	\label{MatYS3}
	&= \underbrace{ \left( \sum_{r=1}^{i-1} \left| \sum_{u \in \mathbb{I}_{0}} {\mathbf{q}_{u}^{H} \mathbf{q}_{u-r}} \right| + \sum_{r=1}^{M-j} \left| \sum_{u \in \mathbb{I}_{0}} {\mathbf{q}_{u}^{H} \mathbf{q}_{u+r}} \right| \right) }_{|G|} \mathbf{D}_{sa},
	\end{flalign}
where $\mathbf{D}_{Y_{c}} \triangleq sv \left( \bar{\mathbf{R}}_{Y_c} \right)$. {{With} the exponential correlation model,} the following theorem {gives} the amplification factor by the mCSLSACC.

\begin{theorem}
\label{thmM}
    Based on the exponential correlation model in \eqref{ECM}, the derived amplification factor of the basic mCSLACC in \eqref{svGain} and the singular value relation in \eqref{MatYS3}, the amplification factor induced by the mCSLSACC over traditional CSL algorithms is given by
    \begin{flalign}
    \label{svGain1}
    \left| G \right| = (j-i+1) \frac{ 2|\rho| - |\rho|^{i} - |\rho|^{M-j+1} }{1-|\rho|}.
    \end{flalign}
\end{theorem}
\begin{proof}
    The proof of \eqref{svGain1} can be found in Appendix \ref{Theorem1}.
\end{proof}
Given a fixed $M$, $|G|$ in \eqref{svGain1} does not always increase with the increasing $j-i+1$. 
The maximum of $|G|$ exists, depending on both $|\rho|$ and $j-i+1$. Accordingly, Theorem \ref{thmM} can be used as a criterion to determine the optimal $j-i+1$ that maximizes $|G|$.

Discussion 1: Theorem \ref{thmM} bridges a connection between the singular value gain and the spatial correlation as well as the number of deployed antennas, providing us a convenient way to analyze how such factors influence the singular values in statistical sense. As the recovery stage retains for the WBSS, the singular values of CSL algorithms should be discriminated from noise for signal subspace extraction.  Considering that the large singular value makes the discrimination accurate, one can use the amplification factor as a guide to design an MIMO system that satisfies the discrimination standard of singular values. For a given SNR that enables traditional CSL algorithms, we can design the spatial correlation and the number of deployed antennas to meet the SNR that enables the discrimination of singular values.
Additionally, given a fixed number of total antennas, the theoretical result provides a criterion to determine the optimal number of deployed antennas when one uses the proposed mCSLSACC algorithm.

\section{Vector Form CSL with Antenna Cross-correlations}
\label{SectionV}
Similarly, we first establish the SCM relations between the proposed vCSLACC with traditional CSL algorithms and then analyze the space diversity and shift factor of the antenna array through the SCM relations.

\subsection{Statistical Covariance Matrix}
In the vCSLACC, the Nyquist samples and sub-samples are reorganized by operator {$\mathcal{A}_{v}$}, which stacks the columns of matrices into vectors. More specifically, defining vectors $\mathbf{n}_{\mathbb{I}_{0}} \triangleq vec ( \mathbf{N}_{\mathbb{I}_{0}} )$, $\mathbf{n}_{\mathbb{I}_{r}} \triangleq vec \left( \mathbf{N}_{\mathbb{I}_{r}} \right)$ and $\mathbf{s} \triangleq vec(\mathbf{S})$, and vectorizing both sides of \eqref{MatRelationship_1} and \eqref{MatRelationship_2}, we obtain
    \begin{flalign}
    \label{VecRelationship2_1}
    \mathbf{y}_{\mathbb{I}_{0}} &= \left( \mathbf{G}_{\mathbb{I}_{0}} \otimes \mathbf{\Omega} \right) \mathbf{s} + \mathbf{n}_{\mathbb{I}_{0}}, \\
    \label{VecRelationship2_2}
    \mathbf{y}_{\mathbb{I}_{r}} &= \left( \mathbf{G}_{\mathbb{I}_{r}} \otimes \mathbf{\Omega} \right) \mathbf{s} + \mathbf{n}_{\mathbb{I}_{r}}.
    \end{flalign}
The SCM $\mathbf{R}(\mathbf{Y},\mathcal{A}_{v},\mathbb{I}_{0},r)$ is thus derived based on \eqref{VecRelationship2_1} and \eqref{VecRelationship2_2} as
	\begin{flalign}
	\nonumber
	\mathbf{R}(\mathbf{Y},\mathcal{A}_{v},\mathbb{I}_{0},r) 
	&= \left( \mathbf{G}_{\mathbb{I}_{0}} \otimes \mathbf{\Omega} \right) \mathbf{R}_{s} \left( \mathbf{G}_{\mathbb{I}_{1}} \otimes \mathbf{\Omega} \right)^{H} + \mathbf{R}(\mathbf{N},\mathcal{A}_{v},\mathbb{I}_{0},r)  \\
	\nonumber
	&= \sum_{k=1}^{K} [\mathbf{g}_{k}]_{\mathbb{I}_{0}} [\mathbf{g}_{k}^{H}]_{\mathbb{I}_{r}} \otimes \mathbf{\Omega} \mathbf{R}_{s_{k}} \mathbf{\Omega}^{H} \\
	\label{VecYS}
	&+ \mathbf{R}(\mathbf{N},\mathcal{A}_{v},\mathbb{I}_{0},r),
	\end{flalign}
where $\mathbf{R}_{\mathbf{s}} \triangleq \mathscr{E}\left\{ \mathbf{s} \mathbf{s}^{H} \right\}$ is a block diagonal matrix, with $\mathbf{R}_{s_{k}} (1 \le k \le K)$ being the diagonal blocks. The second equality of \eqref{VecYS} follows the block diagonal structure of $\mathbf{R}_{\mathbf{s}}$. 
We discuss different cases of the noise term $\mathbf{R}(\mathbf{N},\mathcal{A}_{v},\mathbb{I}_{0},r)$ as follows:
\begin{itemize}
    \item 
    When $0 \le r \le j-i$, the noise term $\mathbf{R}(\mathbf{N},\mathcal{A}_{v},\mathbb{I}_{0},r)$ exists as a block off-diagonal (diagonal when $r=0$) matrix of dimension $(j-i+1)P$, whose off-diagonal blocks are given by $\mathbf{\Pi}_{m} \triangleq diag \{ \alpha_{m} \sigma^{2}_{m}, \cdots, \alpha_{m} \sigma^{2}_{m} \}$ $\in \mathbb{C}^{P \times 1}$ $(1 \le m \le j-i+1-r)$. The SCM $\mathbf{R}(\mathbf{Y},\mathcal{A}_{v},\mathbb{I}_{0},r)$ is therefore corrupted by noise.
    \item
    The noise term $\mathbf{R}(\mathbf{N},\mathcal{A}_{v},\mathbb{I}_{0},r)$ disappears when $j-i < r \le M-j$ and the SCM $\mathbf{R}(\mathbf{Y},\mathcal{A}_{v},\mathbb{I}_{0},r)$ is therefore not corrupted by noise. As a result, the vCSLACC can be used to improve the performance of signal subspace learning in low SNR scenarios.
\end{itemize}

In order to establish the SCM relation between the proposed vCSLACC and traditional CSL algorithms for further analysis, we similarly investigate the SCMs in statistical sense w.r.t. the wireless fading channel as in Section \ref{FOURA}. First of all, we derive $\mathscr{E} \left\{[\mathbf{g}_{k}]_{\mathbb{I}_{0}} [\mathbf{g}_{k}^{H}]_{\mathbb{I}_{r}}\right\}$ as
	\begin{flalign}
	\nonumber
	\mathscr{E}\left\{ [\mathbf{g}_{k}]_{\mathbb{I}_{0}} [\mathbf{g}_{k}^{H}]_{\mathbb{I}_{r}} \right\} 
	&= \sigma_{t_k}^{2} \sum_{m=1}^{M} \left[\mathbf{q}_{m} \right]_{i:j} \left[\mathbf{q}_{m}^{H} \right]_{(i+r):(j+r)} \\
	\label{VecGain}
	=& \sigma_{t_k}^{2} \underbrace{[ ( \mathbf{Q}^{\frac{1}{2}} ) ]_{i:j,:} [ ( \mathbf{Q}^{\frac{1}{2}} )^{H} ]_{:,(i+r):(j+r)} }_{\mathbf{T}_{i,j,r}},
	\end{flalign}
where $\mathbf{T}_{i,j,r}$ is a sub-matrix constituted by rows {of $\mathbf{Q}$} indexed from $i$ to $j$ and columns {of $\mathbf{Q}$} indexed from $i+r$ to $j+r$.
Plugging \eqref{VecGain} into the expectation of \eqref{VecYS}, we then get
	\begin{equation}
	\begin{aligned}
	\bar{\mathbf{R}}(\mathbf{Y},\mathcal{A}_{v},\mathbb{I}_{0},r)
	= \mathbf{T}_{i,j,r} \otimes \bar{\mathbf{R}}_{sa} + \mathbf{R}(\mathbf{N},\mathcal{A}_{v},\mathbb{I}_{0},r),
	\label{VecYS1}
	\end{aligned}
	\end{equation}
with $\bar{\mathbf{R}}(\mathbf{Y},\mathcal{A}_{v},\mathbb{I}_{0},r)  \triangleq \mathscr{E}_{\mathbf{g}_{k}} \left\{\mathbf{R}(\mathbf{Y},\mathcal{A}_{v},\mathbb{I}_{0},r) \right\}$. 
As seen from \eqref{VecYS1}, the amplification factors over singular values of $\bar{\mathbf{R}}_{sa}$ are determined by singular values of matrix $\mathbf{T}_{i,j,r}$. The details will be discussed in the following subsection.

\subsection{Analysis on Space Diversity and Shift Factor}
We also analyze the vCSLACC based on the singular value amplification over traditional CSL algorithms. Since the amplification factors are singular values of $\mathbf{T}_{i,j,r}$ which are not unique usually, we only focus on the {maximum} amplification factor that decides the amplification effect over traditional CSL algorithms. We use this {maximum} amplification factor to determine whether the singular values by the vCSLACC are more easily distinguished from noise than traditional ones. Based on \eqref{VecYS1} and the equation $eig(\mathbf{A} \otimes \mathbf{B}) = eig(\mathbf{A}) \otimes eig(\mathbf{B})$, we obtain the relation of singular values between $\bar{\mathbf{R}}(\mathbf{Y},\mathcal{A}_{v},\mathbb{I}_{0},r)$ and $\bar{\mathbf{R}}_{sa}$ as 
	\begin{equation}
	\begin{aligned}
	\mathbf{D}^{2}(\mathbf{Y},\mathcal{A}_{v},\mathbb{I}_{0},r)
	= \left( \mathbf{D}_{i,j,r}^{t} \right)^{2}
	\otimes \mathbf{D}_{sa}^{2} + \mathbf{D}^{2}(\mathbf{N},\mathcal{A}_{v},\mathbb{I}_{0},r), 
	\label{VecSVRelation}
	\end{aligned}
	\end{equation}
where $\mathbf{D}(\mathbf{Y},\mathcal{A}_{v},\mathbb{I}_{0},r) \triangleq sv \left( \bar{\mathbf{R}}(\mathbf{Y},\mathcal{A}_{v},\mathbb{I}_{0},r) \right)$, $\mathbf{D}_{i,j,r}^{t} \triangleq sv\left( \mathbf{T}_{i,j,r} \right)$ and the diagonal noise matrix $\mathbf{D}^{2}(\mathbf{N},\mathcal{A}_{v},\mathbb{I}_{0},r) = blkdiag\{ \mathbf{\Pi}_{1}^{2}, \cdots, \mathbf{\Pi}_{j-i+1-r}^{2}, \underbrace{ \mathbf{0}, \cdots, \mathbf{0} }_{r} \}$.

When $0 \le r \le j-i$, the noise term exists and the singular value relation \eqref{VecSVRelation} finds an equivalent form as follows 
	\begin{flalign}
	\label{VecEigRelation1}
	\mathbf{D}(\mathbf{Y},\mathcal{A}_{v},\mathbb{I}_{0},r)
	= \mathbf{D}_{i,j,r}^{t}
	\otimes \mathbf{D}_{sa} + \widetilde{\mathbf{D}}(\mathbf{N},\mathcal{A}_{v},\mathbb{I}_{0},r),
	\end{flalign}
where $\widetilde{\mathbf{D}}(\mathbf{N},\mathcal{A}_{v},\mathbb{I}_{0},r)$ denotes the equivalent diagonal noise matrix whose diagonal elements (i.e., the variances of noise) are upper bounded by $\frac{\lambda_{n}^{2}}{\lambda_{y}}$, with $\lambda_{n}$ and $\lambda_{y}$ being the corresponding singular values of $\mathbf{D}(\mathbf{N},\mathcal{A}_{v},\mathbb{I}_{0},r)$ and $\mathbf{D}(\mathbf{Y},\mathcal{A}_{v},\mathbb{I}_{0},r)$, respectively.

To shed more light on the amplification effect, the exponential correlation model \eqref{ECM} is employed. Since the closed form expression of {maximum} singular value is difficult to derive, we characterize the {maximum} singular value with its tight upper and lower bounds. We discuss $r = 0$ and $0 < r \le j-i$, respectively, {in the noisy setup}. 
More specifically, in the case of $r = 0$, we give upper and lower bounds and show the amplification effect over traditional CSL algorithms, whereas, in the case of $0 < r \le j-i$, we only investigate the amplification effect over traditional CSL algorithms.

In the case of $r = 0$, we give {the following} proposition: 
\begin{proposition}
\label{propCase1}
    For $\mathbf{T}_{i,j,r}$ with $r = 0$, the upper and lower bounds on the {maximum} singular value of $\mathbf{T}_{i,j,r}$ are given by
    \begin{flalign}
	\label{lowerbound1}
	B_{lower} &= \frac{1+|\rho|}{1-|\rho|} - \frac{2|\rho| (1-|\rho|^{j-i+1}) }{ (j-i+1) (1-|\rho|)^{2} },\\
	\label{upperbound1}
	B_{upper} &= 
	\left\{ \begin{aligned}
	&\frac{1+|\rho|}{1-|\rho|} - \frac{2|\rho|^{\frac{j-i+2}{2}} }{ 1-|\rho| }&  &\text{for odd} &j-i+1, \\
	&\frac{(1+|\rho|)(1-|\rho|^{\frac{j-i+1}{2}})}{1-|\rho|}& &\text{for even} &j-i+1,
	\end{aligned}  \right.
	\end{flalign}
\end{proposition}
\begin{proof}
	The proofs of \eqref{lowerbound1} and \eqref{upperbound1} can be found in Appendix \ref{Proposition1}.
\end{proof}
It is noted that $B_{lower}$ is no less than one, which reveals that the {maximum} amplification factor has an augmented effect over the singular values of traditional CSL algorithms. Thus, the superiority of vCSLACC {at} $r=0$ over traditional CSL algorithms is concluded and the range of amplification introduced by vCSLACC is given by the bounds.

For $0 < r \le j-i$, we define $s \triangleq j-i-r$ and rewrite $\mathbf{T}_{i,j,r}$ as 
	\begin{equation}
	\begin{aligned}
	&\mathbf{T}_{s,r}
	=\begin{bmatrix}
	(\rho^{*})^{r} &\cdots   &(\rho^{*})^{r+s}  &\cdots & (\rho^{*})^{r+s+r}\\
	\vdots 	   &\ddots   &\vdots 		    &\ddots & \vdots\\
	1 	 	   &\cdots   &(\rho^{*})^{s}	    &\cdots & (\rho^{*})^{r+s}\\
	\vdots 	   &\ddots   &\vdots		    &\ddots & \vdots\\
	\rho^{s}   	   &\cdots   &1			    &\cdots & (\rho^{*})^{r}\\
	\end{bmatrix} \\
	&\triangleq 
	\underbrace{ \left[\begin{aligned} &\mathbf{t}_{1}& &\mathbf{t}_{2}& &\cdots& \mathbf{t}_{s+1} \end{aligned}\right.}_{s+1}
	\underbrace{\left. \begin{aligned} &\rho^{*}\mathbf{t}_{s+1}& &\cdots& (\rho^{*})^{r}\mathbf{t}_{s+1} \end{aligned} \right] }_{r},
	\label{subECM}
	%\vspace{2em}
	\end{aligned}
	\end{equation}
where $\mathbf{t}_{i}, i = 1,2, \cdots, s+1$ are defined {as} columns of $\mathbf{T}_{s,r}$. In this condition, the tight lower and upper bounds of \eqref{subECM} are {too} complicated to {have} {closed {forms}}. Instead, the following theorem provides an alternative approach to show the amplification effect by vCSLACC {for} $0 < r \le j-i$ over traditional CSL algorithms.

\begin{theorem}
\label{thmCase2}
    For $0 < r \le j-i$, the proposed vCSLACC has an augmented effect on singular values over traditional CSL algorithms in statistical sense according to the following inequality
    \begin{flalign}
	\label{inequality1}
	\lambda_{max} \ge \sqrt{ \frac{1}{s+1} trace \left( \mathbf{T}_{s,r}^{H} \mathbf{T}_{s,r} \right) } \ge 1.
	\end{flalign}
\end{theorem}
\begin{proof}
    The proof of \eqref{inequality1} can be found in Appendix \ref{Theorem3}.
\end{proof}
Besides Theorem \ref{thmCase2}, we additionally have an inequality of the derived upper bound on diagonal elements of $\widetilde{\mathbf{D}}(\mathbf{N},\mathcal{A}_{v},\mathbb{I}_{0},r)$ as $\frac{\lambda_{n}^{2}}{\lambda_{y}} \le \lambda_{n}$ according to {\eqref{VecEigRelation1}}, which shows that the noise variance of vCSLACC {for} $0 < r \le j-i$ is lower than that of traditional CSL algorithms. Combining higher singular values with a lower noise variance of the proposed vCSLACC, we conclude that the vCSLACC outperforms traditional CSL algorithms for $0 \le r \le j-i$. It is noted that the noise term exists in this condition and the performance degradation of signal subspace learning still remains.

When $j-i < r \le M-j$, the noise term $\mathbf{D}^{2}(\mathbf{N},\mathcal{A}_{v},\mathbb{I}_{0},r)$ in \eqref{VecSVRelation} disappears, i.e., $\mathbf{D}^{2}(\mathbf{Y},\mathcal{A}_{v},\mathbb{I}_{0},r) = \left( \mathbf{D}_{i,j,r}^{t} \right)^{2} \otimes \mathbf{D}_{sa}^{2}$. Since all matrices in both sides of {such equation} are diagonal, we can {simply write \eqref{VecSVRelation}} as
	\begin{flalign}
	\label{VecEigRelation}
	    \mathbf{D}(\mathbf{Y},\mathcal{A}_{v},\mathbb{I}_{0},r)
    	= \mathbf{D}_{i,j,r}^{t} \otimes \mathbf{D}_{sa}.
	\end{flalign}
 We also employ the exponential correlation model for this case and then $\mathbf{T}_{i,j,r}$ becomes 
	\begin{flalign}
	\nonumber
	\mathbf{T}_{i,j,r} &=
	\begin{bmatrix}
	(\rho^{*})^{r} &(\rho^{*})^{r+1} &\cdots & (\rho^{*})^{r+(j-i)}\\
	(\rho^{*})^{r-1} &(\rho^{*})^{r} &\cdots & (\rho^{*})^{r+(j-i-1)}\\
	\vdots 	   	   &\vdots			 &\ddots & \vdots\\
	(\rho^{*})^{r-(j-i)} &(\rho^{*})^{r-(j-i-1)} &\cdots & (\rho^{*})^{r}
	\end{bmatrix} \\
	\label{subECM1}
	&\triangleq 
	\begin{bmatrix}
	\mathbf{t}_{1} &\rho^{*} \mathbf{t}_{1} &\cdots & (\rho^{*})^{j-i} \mathbf{t}_{1}\\
	\end{bmatrix},
	\end{flalign}
where $\mathbf{t}_{1} \triangleq [(\rho^{*})^{r}, (\rho^{*})^{r-1}, \cdots, (\rho^{*})^{r-(j-i)}]^{T}$. 
It is noted from \eqref{subECM1} that $\mathbf{T}_{i,j,r}$ is a rank-$1$ matrix such that the only singular value of $\mathbf{T}_{i,j,r}$ is also the {maximum} one. 
The following theorem gives the upper and lower bounds on the {maximum} singular value of $\mathbf{T}_{i,j,r}$.

\begin{theorem}
\label{thmCase3}
    For $j-i \le r \le M-j$, the upper and lower bounds on the {maximum} singular value of $\mathbf{T}_{i,j,r}$ are given by %\clearpage
    \begin{flalign}
	\label{lowerbound2}
	B_{lower} &= \frac{|\rho|^{r+1} (1-|\rho|^{j-i+1})(1+|\rho|^{j-i+1})}{|\rho|^{j-i+1}(1-|\rho|)(1+|\rho|)} \\
	\label{upperbound2}
	B_{upper} &= \frac{ |\rho|^{r+1} (1-|\rho|^{j-i+1}) }{ |\rho|^{j-i+1} (1-|\rho|) } \sqrt{ \frac{1+|\rho|^{j-i+1}}{1+|\rho|} },
	\end{flalign}
\end{theorem}
\begin{proof}
    The proofs of \eqref{lowerbound2} and \eqref{upperbound2} can be found in Appendix \ref{Theorem4}.
\end{proof}
{From {\eqref{lowerbound2}-\eqref{upperbound2}}, we know that the bounds {$B_{lower}$ and $B_{upper}$} are functions in terms of $r$, $|\rho|$ and the number of deployed antennas, which makes them convenient to compute.}
Following Theorem \ref{thmCase3}, we obtain an important proposition as follows:
\begin{proposition}
\label{propCase3}
    {For} $j-i \le r \le M-j$, the derived lower bound in \eqref{lowerbound2} is exactly the {maximum} singular value.
\end{proposition}
\begin{proof}
    The proof of this proposition can be found in Appendix \ref{Proposition2}
\end{proof}
With varying $r$ and $|\rho|$, the {maximum} singular value could be smaller or larger than $1$. Thus the vCSLACC {for} $j-i < r \le M-j$ has either amplification or attenuation effects {on the} singular values of traditional CSL algorithms. 

For all previously discussed cases of $r$, an interesting question is how the amplification effect by vCSLACC changes with $r$. The following proposition gives us the answer.

\begin{proposition}
\label{propV}
    For $0 \le r \le M-j$, the amplification factor in statistical sense introduced by the vCSLACC over traditional CSL algorithms decreases with an increasing value of $r$.
\end{proposition}
\begin{proof}
    The proof of this proposition can be found in Appendix \ref{Proposition3}.
\end{proof}

Discussion 2: Similar as in the matrix form algorithm in Section \ref{SectionIV}, the derived amplification bounds are also functions of spatial correlation and the number of deployed antennas, which provides us a way to analyze how such factors influence the singular values in statistical sense, as well as a guide to design an MIMO system meeting the discrimination standard of singular values. Besides, the amplification bounds can be used to determine the performance of spatially correlated MIMO systems (when an exact extreme singular value is unavailable \cite{Choi2014Bounds,Lim2017Bounds,Gong2019ImprovedUB}), as well as to decide the selection of the optimal subarray for wideband hybrid precoding \cite{Park2017Dynamic}.

%===================================================================%
%------------------------Simulation Results-------------------------%
%===================================================================%
\section{Simulation Results}
\label{SectionVII}
In this section, we perform simulations to verify the proposed CSL algorithms from the perspective of the singular value amplification and the WBSS performance.

\subsection{Simulation Settings}
\label{SectionVII-A}
In the simulations, one CR senses a $1$GHz baseband spectrum occupied by three active PUs ($K = 3$), except the simulations in Fig. \ref{SpasityCmp_sim}, in which we analyze the sensing performance for different numbers of PUs. The PU signals are modulated in BPSK signals, and each BPSK signal is transmitted with a unit power and a bandwidth of $20$MHz. The spatially correlated MIMO channel is realized by the Kronecker model with a randomly generated diagonal transmitting correlation matrix and a Toeplitz receiving correlation matrix $\mathbf{Q} = toeplitz(1,\rho, \cdots,\rho^{M-1})$. 
In the simulations, we will show how amplification factors change with different values of $|\rho|$ and $j-i+1$, respectively. Specifically, we vary $|\rho|$ from $0$ to $1.0$, setting $M = 6$, $i = 2$, $j = 3, 4$ for the mCSLSACC and $M = 10$, $i = 2$, $j = 4$ for the vCSLACC. Besides, we vary $j-i+1$ from $1$ to $10$ by fixing $i = 2$ and varying $j$ from $2$ to $11$, and the corresponding $M$ of mCSLSACC and vCSLACC are consequently $M = 11$ and $M = 21$, respectively. In the simulations of WBSS, all algorithms are shown with $M = 6$, $i = 2$, $j = 3$ and $|\rho| = 0.6$.
The random demodulator \cite{Tropp2009Beyond} is used as the sub-Nyquist sampling system and sensing performances under different compression ratios are shown. In addition, for all simulations, we focus on a common sensing period during which channel gains are fixed. Simulation results are averaged over $5000$ Monte Carlo realizations. For each realization, random PU signals, the MIMO channel and noise are generated.

%%======================================================================
\begin{figure}[t!]
	\centering
	\includegraphics[height=14.5cm, width=8.6cm]{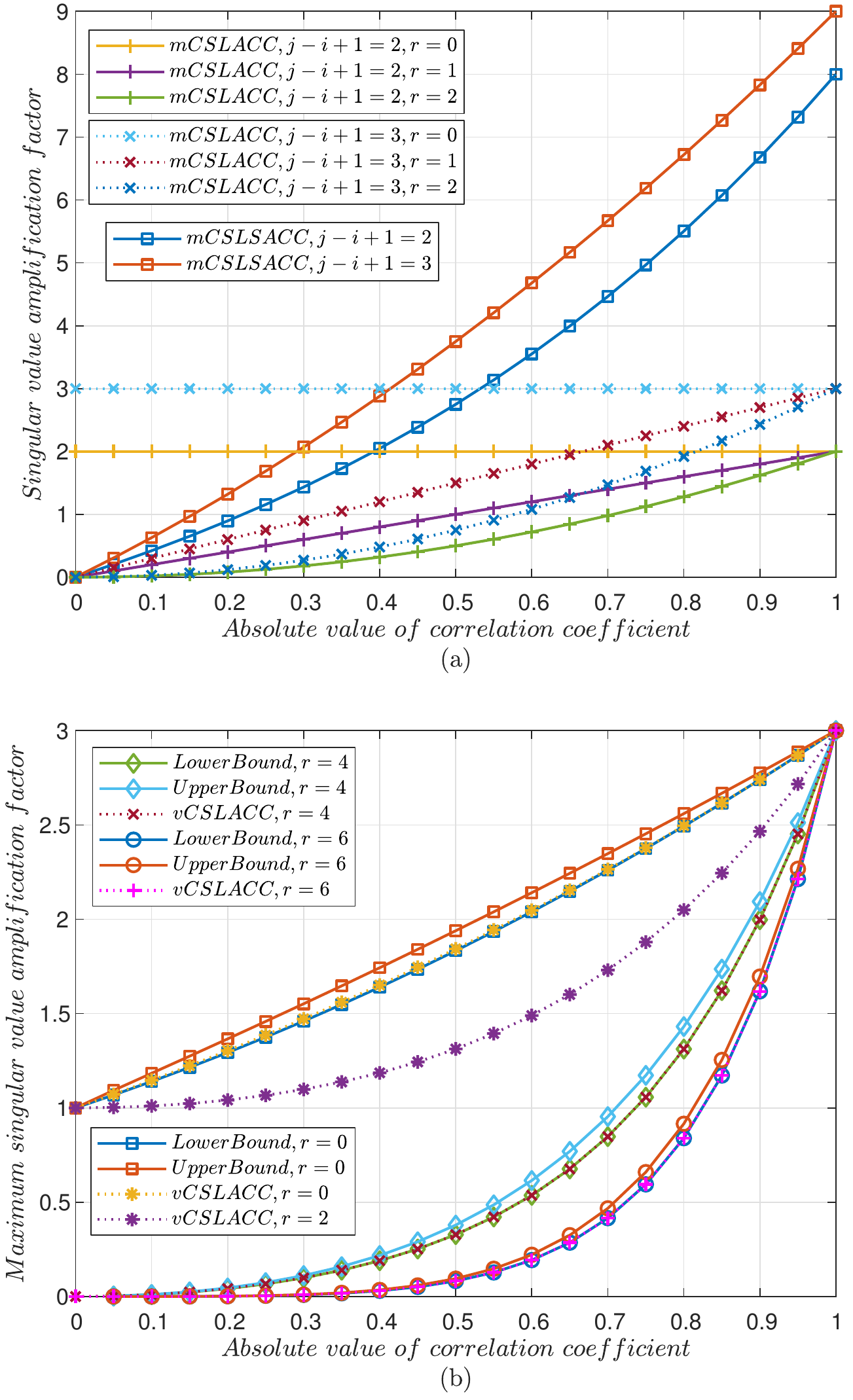}
	\vspace{-0.6em}
	\caption{The curves on how the amplification factor evolves with the absolute value of correlation coefficient $|\rho|$. (a) Comparison between mCSLACC and mCSLSACC. (b) The lower and upper bounds of vCSLACC.}
	\label{mCSLACC_vCSLACC_Gains}
\end{figure}

\begin{figure}[!t]
	\centering
	\includegraphics[height=14.5cm, width=8.6cm]{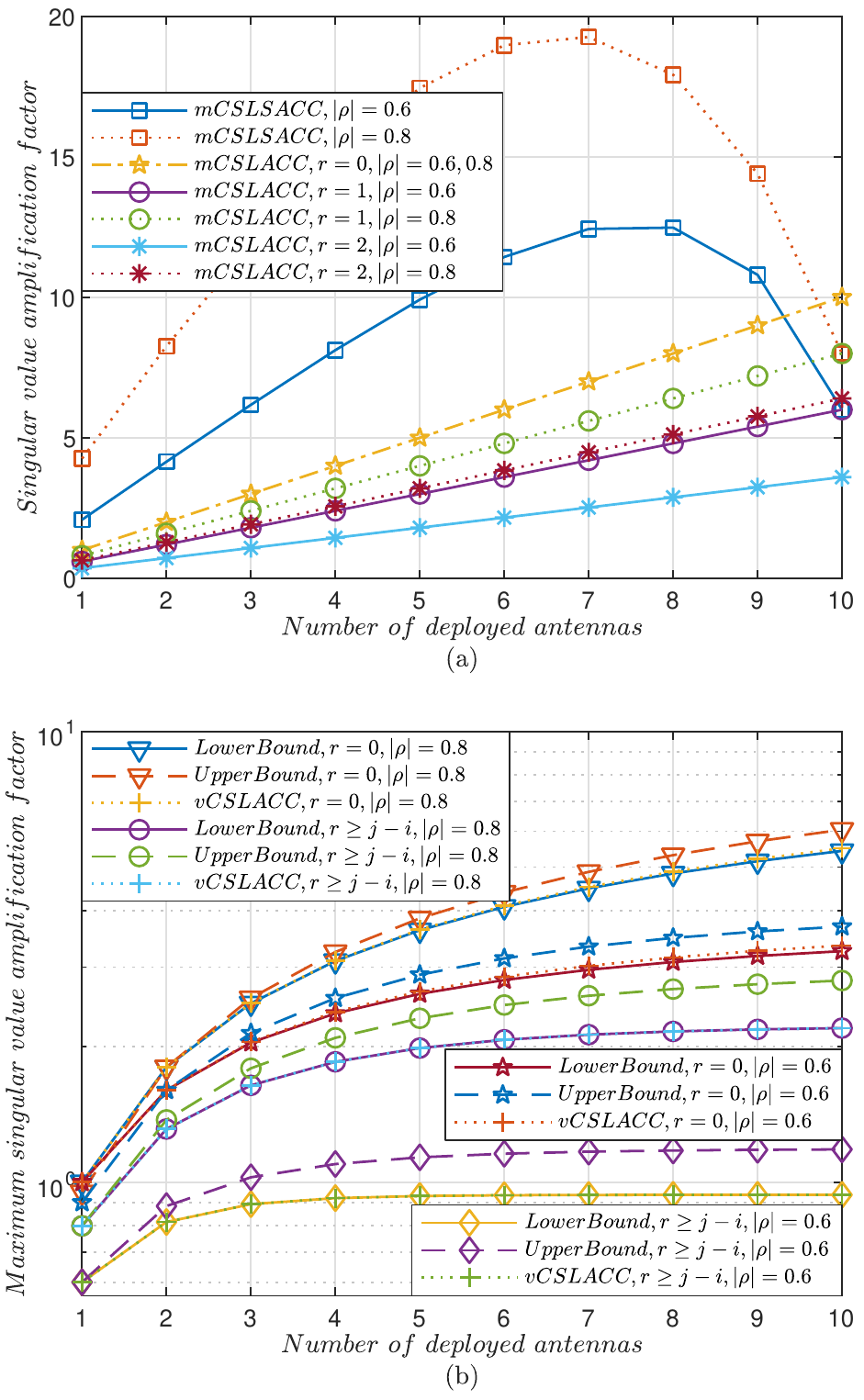}
	\vspace{-0.6em}
	\caption{The curves on how the amplification factor evolves with the number of deployed antennas $j-i+1$. (a) Comparison between mCSLACC and mCSLSACC. (b) The lower and upper bounds of vCSLACC.}
	\label{SVA2AN}
\end{figure}

\subsection{Singular Value Amplification}
\label{SectionVII-B}
The curves on how the amplification factor evolves with the absolute value of correlation coefficient $|\rho|$ are shown in Fig. \ref{mCSLACC_vCSLACC_Gains}(a) (for the mCSLSACC and its basic part mCSLACC) and Fig. \ref{mCSLACC_vCSLACC_Gains}(b) (for the vCSLACC). As seen from Fig. \ref{mCSLACC_vCSLACC_Gains}(a), for a fixed number of deployed antennas (i.e., $j-i+1 = 2, 3$), the amplification factor is attenuated by $|\rho|$ with $r$ varying from $0$ to $2$ and the {maximum} amplification factor is achieved at $r = 0$. It can also be seen that, for a given $r$, the large number of deployed antennas introduces a high amplification factor. In addition, the mCSLSACC yields a larger amplification factor than the basic mCSLACC and the amplification factor increases with $j-i+1$ and $|\rho|$.
In Fig. \ref{mCSLACC_vCSLACC_Gains}(b), the {maximum} amplification factors by vCSLACC with $r = 0,4,6$ are bounded by the corresponding upper and lower bounds. $r = 4,6$ satisfy the condition of the vCSLACC without noise corruption (i.e., $j-i < r \le M-j$), whose curves clearly show that the {maximum} amplification factors exactly reach the corresponding lower bounds. Moreover, we can also observe that the {maximum} amplification factors are monotonic decreasing with $r$.

The curves on how the amplification factor evolves with the number of deployed antennas are shown in Fig. \ref{SVA2AN}. Two values of $|\rho|$ (i.e., $0.6, 0.8$) are selected to denote the medium and high correlations, respectively. In Fig. \ref{SVA2AN}(a), it is shown that the amplification factors by mCSLACC monotonically increase with a growing number of deployed antennas. However, for the mCSLSACC, each amplification factor in terms of $|\rho|$ first increases to a peak point, after which the amplification factor monotonically decreases to the smallest value (i.e., the amplification factor of mCSLACC with $r = 1$). This phenomenon is due to the fact that the mCSLSACC finally reduces to the mCSLACC with $r = 1$ when the number of deployed antenna $j-i+1$ equals $M-1$. As seen from the two curves of mCSLSACC, the number of deployed antennas achieving the peak points tends to be small while $|\rho|$ increases. In Fig. \ref{SVA2AN}(b), the amplification factors by vCSLACC in different cases of $r$ all increase with the number of deployed antennas. In addition, the amplification factors in all cases of $r$ are bounded by the corresponding lower and upper bounds. Particularly, the amplification factors by the vCSLACC with $r \ge j - i$ exactly reach the lower bounds.

\subsection{Wideband Spectrum Sensing}
\label{SectionVII-C}
In this part, the proposed and traditional CSL algorithms are compared in terms of the WBSS performance. They are the mCSLSACC, the basic mCSLACC with $r = 0, 1$ and the vCSLACC with $r = 0, 1, 2$. Traditional CSL algorithms in the MIMO channel without considering the spatial correlations (referred to as the tmaCSL) and traditional CSL algorithms for single antenna scenario (termed as the tsaCSL) are selected as compared algorithms. The probability of detection is adopted as the sensing metric for comparisons. The probabilities of detection for different values of compression rations, number of PUs and SNR levels are displayed in Figs. \ref{CompRatioCmp_sim}, \ref{SpasityCmp_sim} and \ref{WBSS_sim}, respectively. 

\begin{figure}[!t]
	\centering
	\includegraphics[height=6.5cm, width=8.5cm]{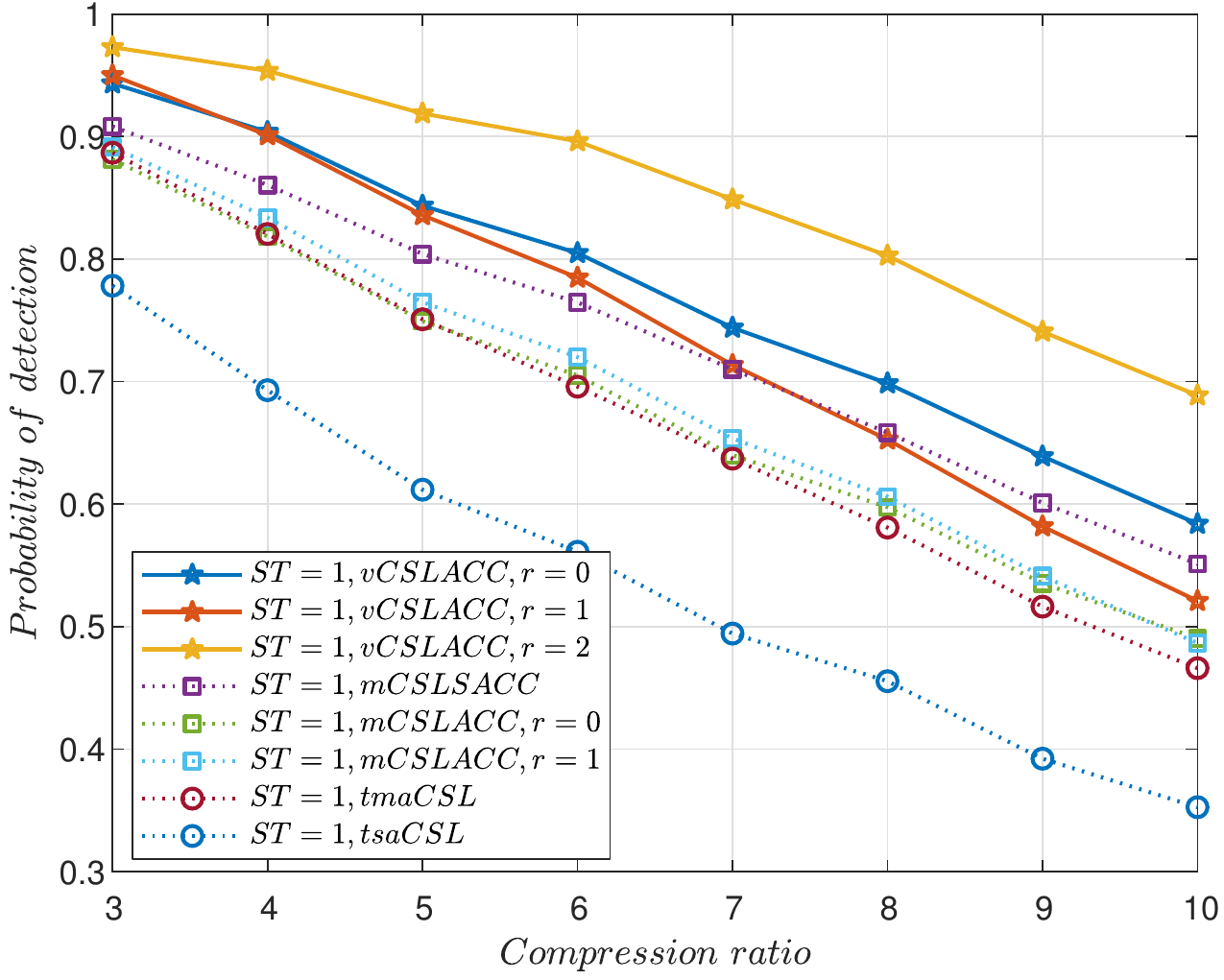}
	\vspace{-0.6em}
	\caption{Performance comparison of WBSS in terms of the compression ratio with $L = 100$ and $Q = 300$.}
	\label{CompRatioCmp_sim}
\end{figure}

In Fig. \ref{CompRatioCmp_sim}, the compression ratio varies from 3 to 10, in other words, the sub-Nyquist sampling rate varies from $1/10$ to $1/3$ of the Nyquist sampling rate. In this simulation, the SNRs of all algorithms are set as $-16$dB and the sensing period is fixed to one unit, i.e., ST = $1$. In Fig. \ref{CompRatioCmp_sim}, large compression ratio degrades the probability of detection, and the proposed vCSLACC, mCSLSACC and the basic mCSLACC achieve higher probabilities of detection than compared algorithms (tmaCSL and tsaCSL). Among the proposed algorithms, the CSL algorithms without noise corruption (i.e., vCSLACC with $r = 2$ and mCSLACC with $r = 1$) outperform those corrupted by noise (i.e., vCSLACC with $r = 0, 1$ and mCSLACC with $r = 0$). The mCSLSACC exploiting the sum of basic mCSLACCs without noise corruption shows the performance improvement over the mCSLACC with $r = 1$.

In Fig. \ref{SpasityCmp_sim}, the number of PUs varies from $4$ to $32$. The SNRs of all algorithms are set as $-18$dB and the sensing period is ST = $1$. It is clearly shown in Fig. \ref{SpasityCmp_sim} that the probabilities of detection of all algorithms decrease when the number of PUs increases and the sensing performances of the proposed algorithms are better than those of traditional ones. In addition, the vCSLACC with $r = 2$, the mCSLACC with $r = 1$ and the mCSLSACC (all are free of noise corruption) are superior to the corresponding proposed algorithms in noise corrupted cases.

\begin{figure}[!t]
	\centering
	\includegraphics[height=6.5cm, width=8.5cm]{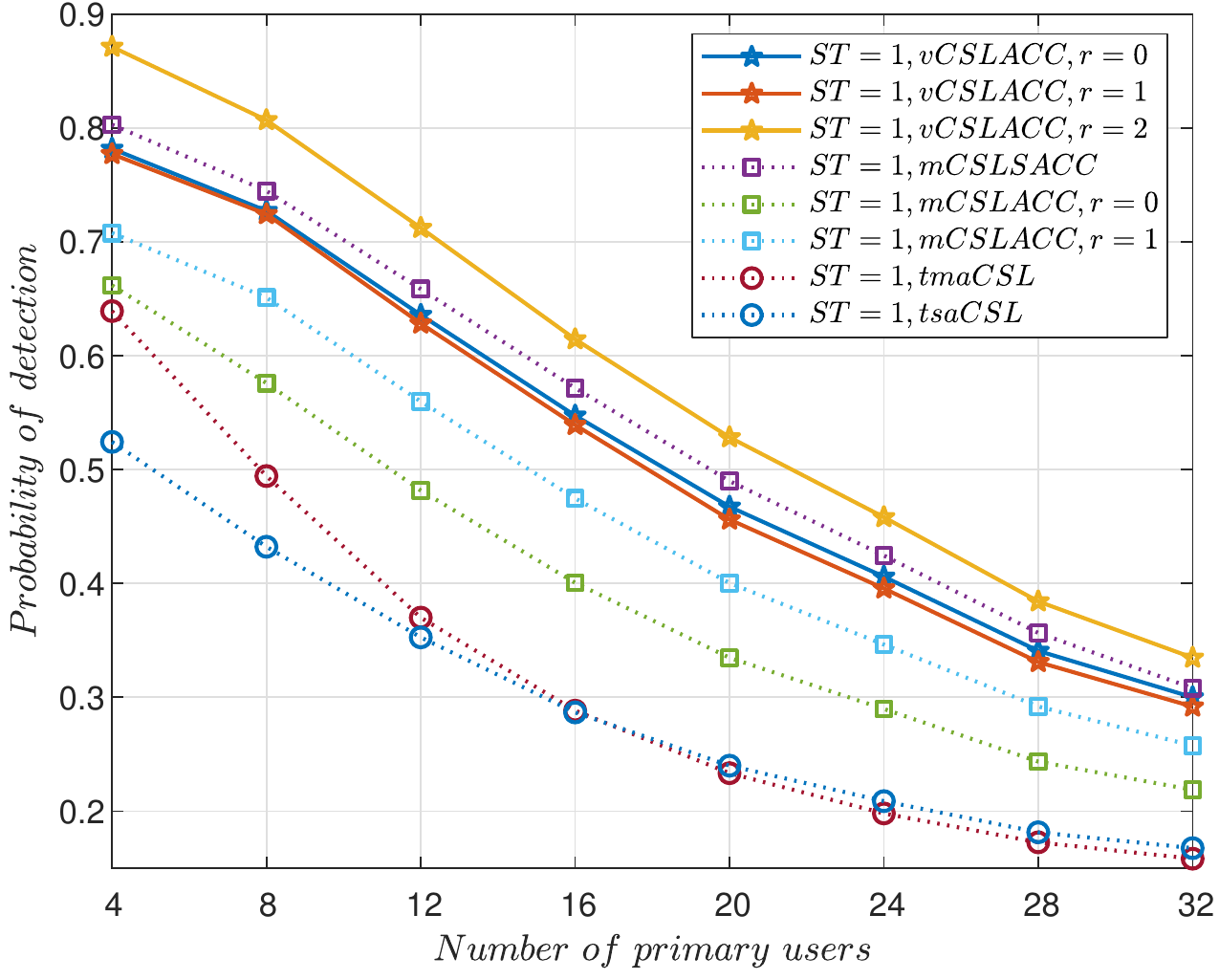}
	\vspace{-0.6em}
	\caption{Performance comparison of WBSS in terms of the PU numbers with $L = 100$, $P = 100$ and $Q = 500$.}
	\label{SpasityCmp_sim}
\end{figure}

\begin{figure}[!t]
	\centering
	\includegraphics[height=6.5cm, width=8.5cm]{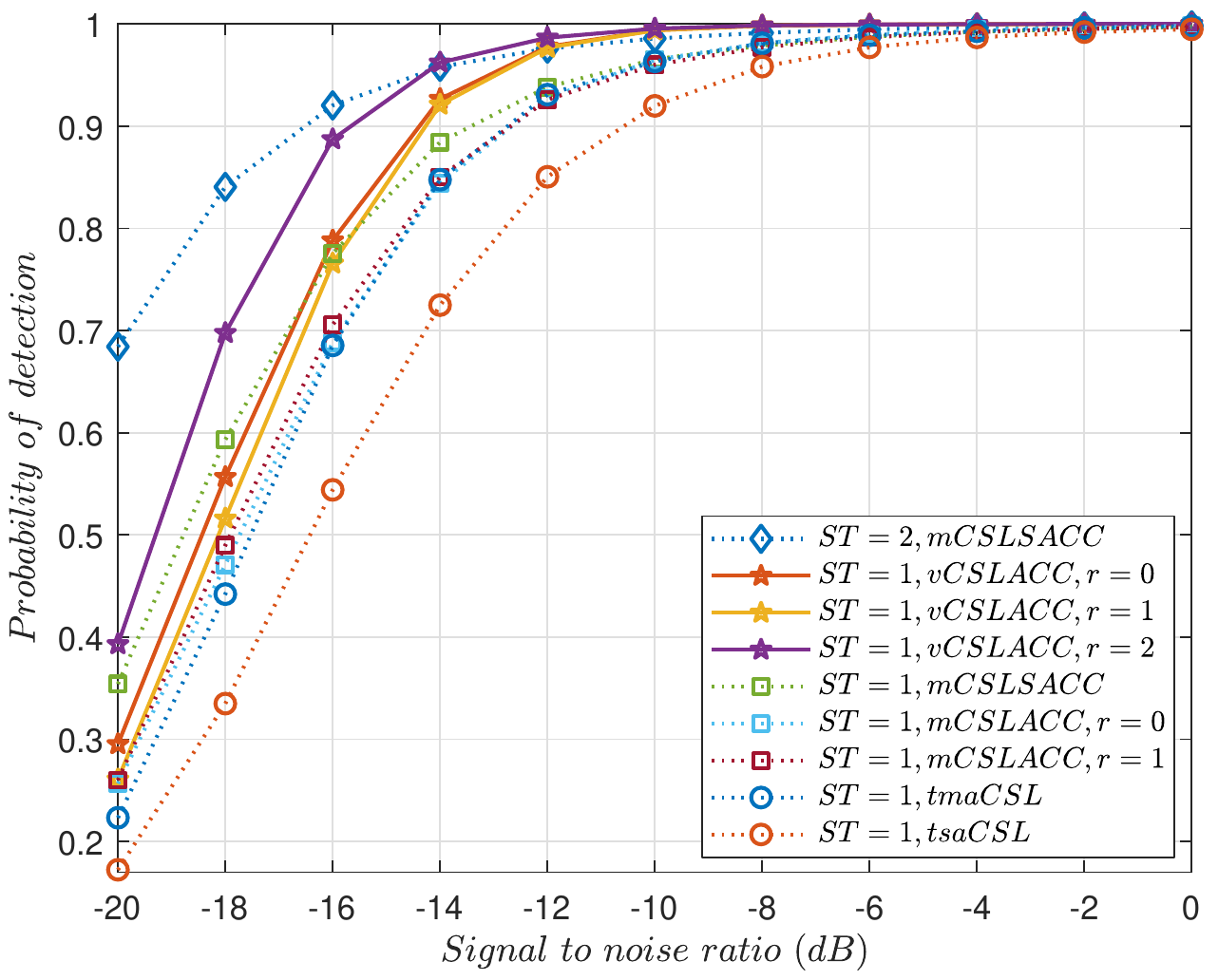}
	\vspace{-0.6em}
	\caption{Performance comparison of WBSS with $L = 100$, $P = 100$, $Q = 200$ (ST = $1$) and with $L = 100$, $P = 200$, $Q = 400$ (ST = $2$).}
	\label{WBSS_sim}
\end{figure}

In Fig. \ref{WBSS_sim}, the SNR varies from $-20$dB to $0$dB. 
Among all solid-line curves, the vCSLACC with $r = 0$ and $r = 1$ corrupted by noise achieve lower detection probabilities than the vCSLACC with $r = 2$ (without noise corruption). Among all dash circle lines in the condition of ST = 1 (or ST = 2), the mCSLSACC performs better than both the mCSLACC with $r = 1$ (without noise corruption) and the mCSLACC with $r = 0$ (corrupted by noise). 
In addition, for ST = $1$, both noise corrupted and noise free cases of the vCSLACC outperform the corresponding noise corrupted and noise free cases of the mCSLACC. We increase the sensing period of the mCSLSACC and the mCSLACC (i.e., ST = $2$) in order to keep a unified SCM dimension with the vCSLACC (i.e., ST = $1$).
Then it can be seen that the mCSLSACC and the mCSLACC outperform the vCSLACC with $r = 2$ in low SNR scenarios (below $-14$dB/$-16$dB for the mCSLSACC/mCSLACC). In contrast, in high SNR scenarios, the vCSLACC shows its advantage over the mCSLSACC and the mCSLACC.
Finally, compared with tmaCSL and tsaCSL, the proposed CSL algorithms achieve higher probabilities of detection.

%===================================================================%
%---------------------------Conclusion------------------------------%
%===================================================================%
\section{Conclusion}
\label{SectionVIII}
In this paper, the antenna cross-correlations and the MIMO spatial correlation have been introduced to promote the compressive subspace learning. To this end, two CSL algorithms (mCSLSACC and vCSLACC) have been proposed.
Based on the SCM statistical relation between the proposed and traditional CSL algorithms, space diversity and shift factor of antennas have been analyzed for both proposed algorithms based on the Kronecker channel model. By further using the exponential correlation model as the receiving correlation matrix, useful closed form expressions and conclusions for proposed algorithms have been derived.
Simulations have verified the derived results and shown performance improvements in WBSS over traditional CSL algorithms.

% or

%\vfill\pagebreak
%===================================================================%
%---------------------------Appendices------------------------------%
%===================================================================%
\appendices
%----------------------------------------------------------------------------------------

\section{Proof of Theorem 1}
\label{Theorem1}
    Assuming that three antennas labeled by $m_1$, $m_2$ and $m_3$ are with equally separated spaces, we have $\rho_{m_1 m_2} = \rho_{m_2 m_3}$. With the symmetric property of antenna correlations, we obtain that $\left| \sum_{u \in \mathbb{I}_{0}} {\mathbf{q}_{u}^{H} \mathbf{q}_{u-r}} \right| = \left| \sum_{u \in \mathbb{I}_{0}} {\mathbf{q}_{u}^{H} \mathbf{q}_{u+r}} \right|$. Therefore, the amplification factor in \eqref{svGain1} is obtained by plugging \eqref{svGain} into $|G|$ in \eqref{MatYS3} and then calculating a geometric sequence.

\section{Proof of Proposition 1}
\label{Proposition1}
    In the case of $r = 0$, the matrix $\mathbf{T}_{i,j,r}$ is also an exponential correlation matrix given by a sub-matrix of \eqref{ECM} with both columns and rows indexed by $\mathbb{I}_{0} = \{i,i+1,\cdots,j\}$. Since the {maximum} singular value and the {maximum} eigenvalue of $\mathbf{T}_{i,j,r}$ are equal, the lower bound in \eqref{lowerbound1} and the upper bound in \eqref{upperbound1} on the {maximum} singular value of $\mathbf{T}_{i,j,r}$ directly follow the results in \cite{Choi2014Bounds,Park2017Dynamic} and \cite{Lim2017Bounds}, respectively.

\section{Proof of Theorem 2}
\label{Theorem3}
    As seen from the structure of $\mathbf{T}_{s,r}$ in \eqref{subECM}, we know that, between any two different columns, the first $s+1$ columns of $\mathbf{T}_{s,r}$ are linearly independent but the last $r+1$ columns are linearly dependent. Therefore, the number of eigenvalues of $\mathbf{T}_{s,r}^{H} \mathbf{T}_{s,r}$ is $s+1$, leading to an average value of all eigenvalues as 
    \begin{flalign}
	\nonumber
	\bar{C} 
	&\triangleq \frac{1}{s+1} trace \left( \mathbf{T}_{s,r}^{H} \mathbf{T}_{s,r} \right) \\
	\label{aveEig}
	&= \frac{1}{s+1} \left( \sum\limits_{k=1}^{s+1} \mathbf{t}_{k}^{H} \mathbf{t}_{k} + \sum\limits_{k=1}^{r} |\rho|^{2k} \mathbf{t}_{s+1}^{H} \mathbf{t}_{s+1} \right).
	\end{flalign}
    Afterwards, we have the following inequalities
    \begin{flalign}
	\label{inequality1_1}
	\bar{C}
	\ge \frac{1}{s+1} \sum\limits_{k=1}^{s+1} \mathbf{t}_{k}^{H} \mathbf{t}_{k}
	\ge 1,
	\end{flalign}
    where the equalities of \eqref{inequality1_1} hold if and only if $\rho = 0$. Finally, we conclude \eqref{inequality1} due to the fact that $\lambda_{max} \ge \sqrt{\bar{C}}$.

\section{Proof of Theorem 3}
\label{Theorem4}
%\clearpage
    By using the rank-$1$ $\mathbf{T}_{i,j,r}$ in \eqref{subECM1}, we know that the {maximum} singular value of $\mathbf{T}_{i,j,r}$ is larger than the {maximum} square root of diagonal elements of $\mathbf{T}_{i,j,r}^{H} \mathbf{T}_{i,j,r}$ as the following inequality
	\begin{flalign}
	\label{inequality2}
	\lambda_{max} 
	\ge \sqrt{ \sum\limits_{k = 0}^{j-i} |\rho|^{2k} \mathbf{t}_{1}^{H} \mathbf{t}_{1} } 
	\ge \max_{0 \le k \le j-i} \sqrt{ |\rho|^{2k} \mathbf{t}_{1}^{H} \mathbf{t}_{1} }.
	\end{flalign}
    Moreover, combining \eqref{inequality2} with the Gershgorin's circle theorem \cite{George2007Handbook}, the {maximum} singular value is upper bounded by the square root of the {maximum} row sum of a matrix with its elements as $\left| [\mathbf{T}_{i,j,r}^{H} \mathbf{T}_{i,j,r}]_{m,n} \right| (1 \le m, n \le j-i+1)$. We derive the upper bound as
	\begin{flalign}
	\nonumber
	B_{upper} 
	&= \sqrt{\sum\limits_{k = 0}^{j-i} |\rho|^{k} \mathbf{t}_{1}^{H} \mathbf{t}_{1}} 
	= \sqrt{\sum\limits_{k = 0}^{j-i} |\rho|^{k} \sum\limits_{l = r-(j-i)}^{r} |\rho|^{2l}} \\
	\label{upperboundderive}
	&= \frac{ |\rho|^{r+1} (1-|\rho|^{j-i+1}) }{ |\rho|^{j-i+1} (1-|\rho|) } \sqrt{ \frac{1+|\rho|^{j-i+1}}{1+|\rho|} }.
	\end{flalign}
    In addition, the {maximum} singular value of $\mathbf{T}_{i,j,r}$ is lower bounded by the square root of Rayleigh quotient $\sqrt{ \frac{\mathbf{u}^{H} \mathbf{T}_{i,j,r}^{H} \mathbf{T}_{i,j,r} \mathbf{u}}{\mathbf{u}^{H} \mathbf{u}} }$ for any nonzero $\mathbf{u} \in \mathbb{C}^{(j-i+1) \times 1}$. Let $\mathbf{u} = \frac{1}{\sqrt{j-i+1}}[1, \rho, \cdots, \rho^{j-i}]^{T}$ be the non-zero vector, then the lower bound on {maximum} singular value of $\mathbf{T}_{i,j,r}$ is derived as
	\begin{flalign}
	\nonumber
	B_{lower} 
	&= \sqrt{ \frac{\mathbf{u}^{H} \mathbf{T}_{i,j,r}^{H} \mathbf{T}_{i,j,r} \mathbf{u}}{\mathbf{u}^{H} \mathbf{u}} } \\
	\nonumber
	&= \sqrt{ \frac{\mathbf{t}_{1}^{H} \mathbf{t}_{1}}{\mathbf{u}_{1}^{H} \mathbf{u}_{1}} \left( \mathbf{u}_{1}^{H} \mathbf{u}_{1} \sum\limits_{k = 0}^{j-i} |\rho|^{2k} \right)} \\
	\label{EQ_LowerBound}
	&= \frac{|\rho|^{r+1} (1-|\rho|^{j-i+1})(1+|\rho|^{j-i+1})}{|\rho|^{j-i+1}(1-|\rho|)(1+|\rho|)}.
	\end{flalign}

\section{Proof of Proposition 2}
\label{Proposition2}
    We first prove that the nonzero vector $\mathbf{u} = \frac{1}{\sqrt{j-i+1}}[1, \rho, \cdots, \rho^{j-i}]^{T}$ used to derive the lower bound in Appendix \ref{Theorem4} is the eigenvector of $\mathbf{T}_{i,j,r}^{H} \mathbf{T}_{i,j,r}$ through the following equation
	\begin{flalign}
	\nonumber
	\mathbf{T}_{i,j,r}^{H} \mathbf{T}_{i,j,r} \mathbf{u}
	&= \sqrt{j-i+1} \left( \mathbf{t}_{1}^{H} \mathbf{t}_{1} \right) 
	\begin{bmatrix}
	\mathbf{u} &\rho^{*}\mathbf{u} &\cdots & (\rho^{*})^{j-i}\mathbf{u}
	\end{bmatrix} \mathbf{u} \\
	\label{eigenvector}
	&= (j-i+1) \left( \mathbf{t}_{1}^{H} \mathbf{t}_{1} \right)  \left( \mathbf{u}^{H} \mathbf{u} \right) \mathbf{u},
	\end{flalign}
    where the corresponding eigenvalue of $\mathbf{T}_{i,j,r}^{H} \mathbf{T}_{i,j,r}$ is given by 
    \begin{flalign}
	\label{eigenvalue}
	\lambda 
	= (j-i+1) \left( \mathbf{t}_{1}^{H} \mathbf{t}_{1} \right)  \left( \mathbf{u}^{H} \mathbf{u} \right) 
	= \left( \mathbf{t}_{1}^{H} \mathbf{t}_{1} \right) \sum_{k=0}^{j-i} |\rho|^{2k}.
	\end{flalign}
    Based on \eqref{EQ_LowerBound} in Appendix \ref{Theorem4} and \eqref{eigenvalue}, $\lambda = B_{lower}^{2}$ is straightforward and the conclusion of Proposition \ref{propCase3} is proved.

\section{Proof of Proposition 3}
\label{Proposition3}
    Without loss of generality, we take $\mathbf{T}_{s,r}$ in \eqref{subECM} into consideration because both $\mathbf{T}_{i,j,r}$ with $r=0$ and $\mathbf{T}_{i,j,r}$ with $j-i < r \le M-j$ can be considered as two special cases of $\mathbf{T}_{s,r}$. Based on \eqref{subECM} and \eqref{aveEig}, the trace of $\mathbf{T}_{s,r}^{H} \mathbf{T}_{s,r}$ is given by $\bar{C}(s+1)$,
    \iffalse
	\begin{flalign}
	trace(\mathbf{T}_{s,r}^{H} \mathbf{T}_{s,r}) = \sum_{k=1}^{s+1} \mathbf{t}_{k}^{H} \mathbf{t}_{k} + \sum_{k=1}^{r} |\rho|^{2k} \mathbf{t}_{s+1}^{H} \mathbf{t}_{s+1},
	\end{flalign}
	\fi
    which decreases monotonically with $r$. Defining two different variables $r_{1}$, $r_{2}$ $(r_{1} < r_{2})$ and applying the monotonic property of $\bar{C}(s+1)$, we get that the matrix $\mathbf{T}_{s,r_{1}}^{H} \mathbf{T}_{s,r_{1}} - \mathbf{T}_{s,r_{2}}^{H} \mathbf{T}_{s,r_{2}}$ is Hermitian and positive definite because $trace(\mathbf{T}_{s,r_{1}}^{H} \mathbf{T}_{s,r_{1}} - \mathbf{T}_{s,r_{2}}^{H} \mathbf{T}_{s,r_{2}}) > 0$. As a result, the {maximum} eigenvalue (singular value) of $\mathbf{T}_{s,r_{1}}^{H} \mathbf{T}_{s,r_{1}}$ $(\mathbf{T}_{s,r_{1}})$ is larger than that of $\mathbf{T}_{s,r_{2}}^{H} \mathbf{T}_{s,r_{2}}$ $(\mathbf{T}_{s,r_{2}})$ \cite{George2007Handbook}. Therefore, we can make a conclusion that the {maximum} eigenvalue (singular value) of $\mathbf{T}_{s,r}^{H} \mathbf{T}_{s,r}$ $(\mathbf{T}_{s,r})$ decreases with an increasing value of $r$, and the amplification effect introduced by vCSLACC over traditional CSL algorithms diminishes along with an increasing value of $r$.

\bibliographystyle{IEEEtran}
\bibliography{IEEEabrv,references} %逗号前后不可以有空格

\end{document}